\documentclass[10pt, conference, compsocconf]{IEEEtran}

\usepackage{times}
\usepackage{amsmath}
\usepackage{stmaryrd}
\usepackage{amssymb}
\usepackage{graphicx}
\usepackage[T1]{fontenc}

\usepackage[ruled]{algorithm2e}

\newcommand{\tr}{{\rm tr}}

\def\>{\ensuremath{\rangle}}
\def\<{\ensuremath{\langle}}

\newtheorem{thm}{Theorem}[section]

\newtheorem{lem}{Lemma}[section]
\newtheorem{defn}{Definition}[section]
\newtheorem{prop}{Proposition}[section]
\newtheorem{exam}{Example}[section]

\begin{document}
\title{Toward Automatic Verification of Quantum Programs}

\author{\IEEEauthorblockN{Mingsheng Ying}
\IEEEauthorblockA{Centre for Quantum Software and Information, 
University of Technology Sydney,
Australia\\ State Key Laboratory of Computer Science, Institute of Software, Chinese Academy of Sciences, China\\ Department  of Computer Science and Technology, Tsinghua University, China\\
Email: Mingsheng.Ying@uts.edu.au; Yuan.Feng@uts.edu.au}
}

\maketitle

\begin{abstract}
This paper summarises the results obtained by the author and his collaborators in a program logic approach to the verification of quantum programs, including quantum Hoare logic, invariant generation and termination analysis for quantum programs. It also introduces the notion of proof outline and several auxiliary rules for more conveniently reasoning about quantum programs. Some problems for future research are proposed at the end of the paper.
\end{abstract}

\IEEEpeerreviewmaketitle

\section{Introduction}\label{introduce}

Programming is error-prone. Programming a quantum computer and designing quantum communication protocols are even worse due to the weird nature of quantum systems \cite{Ying2016}. Therefore, verification techniques and automatic tools for quantum programs and quantum protocols will be indispensable whence commercial quantum computers and quantum communication systems are available. In the last 10 years, various verification techniques for classical programs have been extended to deal with quantum programs; in particular, several quantum program logics have been proposed: \begin{itemize}\item Brunet and Jorrand~\cite{BJ04} proposed to apply Birkhoff-von Neumann quantum logic in reasoning about quantum programs. 
In the Birkhoff-von Neumann logic, closed subspaces of the state Hilbert space of a quantum system are used to represent the properties of the system, and logical connectives $\neg$ (negation), $\wedge$ (and), $\vee$ (or) are interpreted as the operations: ortho-complement, meet and join, respectively, in the orthomodular lattice of all closed subspaces of the Hilbert space. The idea of~\cite{BJ04} is to add quantum operations (e.g. unitary transformations and quantum measurements) into the language of quantum logic so that it can be employed to describe and reason about the dynamics (evolution) of a quantum system (e.g. a quantum program). 

\item Chadha, Mateus and Sernadas~\cite{CMS} presented one of the first attempts to develop a Hoare-like logic for quantum programs. The quantum programs considered in~\cite{CMS} can have both classical and quantum variables (memories). 
But only bounded iterations are allowed in these quantum programs. 
The logic given in ~\cite{CMS} is an extension of exogenous quantum logic \cite{MS06}, which has terms representing \textit{amplitudes} of quantum states, obtained by introducing the axioms for unitary transformations and quantum measurements. Roughly speaking, the assertions (preconditions and postconditions) in this logic are properties of the real and imaginary components of the amplitudes. 

\item Baltag and Smets~\cite{BS06} chose to develop a quantum generalisation of Propositional Dynamic Logic (PDL), which is an extension of Hoare logic for verification of classical programs. The approach of \cite{BS06} is different from \cite{BJ04}, \cite{CMS}, and it can be essentially seen as a reinterpretation of the language of PDL where a quantum action is interpreted as a (classical) relation between quantum states and models the input-output relation of a quantum program.
In particular, by introducing some axioms to handle separation and locality, their proof system can be used to reason about the behaviours (e.g. entanglement) of compound quantum systems.

\item Kakutani \cite{Kaku09} proposed a quantum extension of Den Hartog and De Vink's probabilistic Hoare logic \cite{HV02}. The quantum programming language used in \cite{Kaku09} is Selinger's QPL \cite{Selinger04} without recursive calls but with while loops.  The assertions in this logic are probabilistic predicates (or assertions) used in \cite{HV02} but, of course, re-interpreted in the setting of quantum computation (e.g. the probability that qubit variable $q$ returns to $0$ upon the measurement in the computational basis is $\frac{1}{4}$).  

\item Some useful proof rules for reasoning about quantum programs were introduced by Feng et al.~\cite{FDJY07}. The programming language employed in~\cite{FDJY07} is a natural quantum extension of the \textbf{while}-language. For simplicity of presentation, only purely quantum programs without classical variables were considered there. The assertions in these proof rules are physical observables (mathematically modelled by Hermitian operators), or quantum predicates as called in \cite{DP06}.   
Furthermore, a Hoare-like logic for both partial and total correctness of quantum programs with (relative) completeness was established in~\cite{Ying2011}. 
\end{itemize}
The mathematic foundations, the programming languages and the expressivity of the assertions of the quantum program logics presented in \cite{CMS}, \cite{Kaku09}, \cite{Ying2011} are 
carefully compared in a recent survey \cite{Rand17}.

Except program logics, model-checking techniques have also been developed for verification of quantum communication protocols and quantum programs. For example, Gay, Papanikolaou and Nagarajan \cite{GPN08} developed a model-checker for verifying properties of the quantum systems that can be modelled in the stabiliser formalism. Feng et al. \cite{FengYY13} presented an algorithm for model-checking quantum systems described as super-operator valued Markov chains, and the algorithm was implemented by Feng et al. in \cite{FengHTZ15}. The problem of checking linear-time properties of quantum systems was studied in \cite{TOCL13}, where linear-time properties are defined to be infinite sequences of sets of atomic propositions modelled as closed subspaces of the state Hilbert spaces, as in the Birkhoff-von Neumann quantum logic.   

The main purpose of this paper is to summarise the results obtained by the author and his collaborators in a program logic approach to the verification of quantum programs. The paper is an extension of an invited talk at SETTA'2016 \cite{Ying2016a}, but the results presented in Sections \ref{outline} and \ref{auxiliary} as well as the examples given in Section \ref{sec-examples} are new and have not been published elsewhere. 

The paper is organised as follows: The syntax and semantics of a quantum extension of the \textbf{while}-language is defined and a logic of Hoare-style for reasoning about quantum programs is presented in Section \ref{qwhile}. 
As is well-known, correctness specification of classical programs can be very tricky. Indeed, correctness specification of quantum programs can be even much trickier. In Section \ref{sec-examples}, we present several simple examples through which the reader can better understand the difference between classical and quantum correctness specifications. The notion of proof outline is introduced and a strong soundness theorem is proved for quantum programs in Section \ref{outline}. Several auxiliary axioms and inference rules are given in Section \ref{auxiliary}, which can help to simplify verifications of quantum programs. Some basic ideas for mechanising quantum program verification based on the quantum Hoare logic are described in Section \ref{mechanising}. The algorithms for generating invariants and termination analysis of quantum programs are summarised in Sections \ref{invariant} and \ref{termination}, respectively. Some topics for future studies are pointed out in the concluding section. For convenience of the reader, several basic properties of operators in Hilbert spaces needed in this paper are listed in Appendix \ref{app-a}. Several simple quantum relations (i.e. quantum predicates in a tensor product of state Hilbert spaces) were employed in Section \ref{sec-examples}, a brief discussion on composition of quantum relations that can be used to generate more sophisticated quantum relations from these simple ones is given in Appendix \ref{app-relation}. For readability, the proofs of the results in Sections \ref{outline} and \ref{auxiliary} are deferred to Appendix \ref{app-b}.  

\section{Hoare Logic for Quantum Programs}\label{qwhile}

Hoare logic is a cornerstone in verification of classical programs. As mentioned in Section \ref{introduce}, several attempts \cite{BJ04}, \cite{CMS}, \cite{BS06}, \cite{FDJY07}, \cite{Kaku09}, \cite{Ying2011} have been made to build a Hoare-like logic for quantum programs. In this section, we focus on the one presented in \cite{Ying2011} because it is the only quantum Hoare logic with (relative) completeness in the existing literature.  

\subsection{Syntax and Semantics}

Let us first briefly recall the syntax of a quantum extension of the \textbf{while}-language. 
It is assumed that the reader is familiar with the basic notions of quantum computing. (A reader not familiar with them can have a quick look at Section 2 of \cite{Ying2011} or Section 1.1 of \cite{YYW17}; for more details, we refer to \cite{Nielsen} or \cite{Ying2016}.) We assume a countably infinite set $\mathit{Var}$ of quantum variables. For each $q\in\mathit{Var}$, we write $\mathcal{H}_q$ for its state Hilbert space. 

\begin{defn}[Syntax \cite{Ying2011}, \cite{Ying2016}]\label{q-syntax} 
The quantum \textbf{while}-programs are defined by
the grammar:
\begin{align}\label{syntax}P::=\ \mathbf{skip}\ & |\ P_1;P_2\ |\ q:=|0\rangle\ |\ \overline{q}:=U[\overline{q}]\\ \label{syntax+}&|\ \mathbf{if}\ \left(\square m\cdot M[\overline{q}] =m\rightarrow P_m\right)\ \mathbf{fi}
\\ \label{syntax++}&|\ \mathbf{while}\ M[\overline{q}]=1\ \mathbf{do}\ P\ \mathbf{od}\end{align}\end{defn}

The command \textquotedblleft$q:=|0\rangle$\textquotedblright\ is an initialisation that sets quantum variable $q$ to a basis state $|0\rangle$.
 The statement \textquotedblleft$\overline{q}:=U[\overline{q}]$\textquotedblright\ means that unitary transformation $U$ is performed on quantum
register $\overline{q}$, leaving the states of the variables
not in $\overline{q}$ unchanged.
The construct \textquotedblleft$\mathbf{if}\cdots\mathbf{fi}$\textquotedblright\ is a quantum
generalisation of case or switch statement. In executing it, measurement $M=\{M_{m}\}$ is performed on 
$\overline{q}$, and then a subprogram $P_m$ is selected to be executed next according to the outcomes $m$ of
measurement. The statement \textquotedblleft$\mathbf{while}\cdots\mathbf{od}$\textquotedblright\ is a quantum generalisation of \textbf{while}-loop. The measurement in it has only two possible outcomes $0$, $1$. If the outcome $0$ is observed, then the program terminates, and if the outcome $1$ occurs, the
program executes the loop body $P$ and continues the loop.

To further illustrate these program constructs, let us see two simple examples, both of which are quantum generalisations of certain probabilistic programs. The first is taken from \cite{YYW17}, and it is the quantum version of the example of three dials in \cite{KMMM10}.

\begin{exam}[Three Quantum Dials]\label{exam0} Consider a slot machine that has three dials $d_1,d_2,d_3$ and two suits $\heartsuit$ and $\diamondsuit$. It spins the dials independently so that they come to rest on each of the suits with equal probability. This machine can be modelled as a probabilistic program: \begin{align*}\mathit{flip}\equiv\ (d_1:=\heartsuit \oplus_{\frac{1}{2}} d_1:=\diamondsuit);\ &(d_2:=\heartsuit \oplus_{\frac{1}{2}} d_2:=\diamondsuit);\\ &(d_3:=\heartsuit \oplus_{\frac{1}{2}} d_3:=\diamondsuit)\end{align*} where $P_1 \oplus_{p} P_2$ stands for a probabilistic choice which chooses to execute $P$ with probability $p$ and to execute $Q$ with probability $1-p$.

We can define a quantum variant of $\mathit{flip}$ as follows:
$$\mathit{qflip}\equiv H[d_1];\ H[d_2];\ H[d_3]$$ where $H$ is the Hadamard gate in the $2$-dimensional Hilbert space $\mathcal{H}_2$ with $\{|\heartsuit\rangle,|\diamondsuit\rangle\}$ as an orthonormal basis. The program $\mathit{qflip}$ also spins the dials, but does it in a quantum way modelled by the Hadamard \textquotedblleft coin-tossing\textquotedblright\ operator $H$. 

It is worth pointing out an interesting difference between the probabilistic and quantum setting: spinning the dials with equal probability can be implemented in many different ways in the quantum world; e.g. a quantum gate different from the Hadamard gate. 
\end{exam} 

The second example is a quantum generalisation of random walk that every one working on probabilistic algorithms or programming must be familiar with. Here is a simplified version (taken from \cite{Yi13}) of the one-dimensional quantum walks defined in \cite{Am-B01}.

\begin{exam}[Quantum Walk]\label{exam1} Let $\mathcal{H}_c$ be the $2$-dimensional Hilbert space with orthonormal basis states $|L\rangle$ and $|R\rangle$, indicating directions Left and Right, respectively. It is the state space of a quantum coin. Let $\mathcal{H}_p$ be the $n$-dimensional Hilbert space with orthonormal basis states $|0\rangle, |1\rangle, ..., |n-1\rangle$, where vector $|i\rangle$ denotes position $i$ for each $0\leq i< n$. The state space of the walk is $\mathcal{H}=\mathcal{H}_c\otimes \mathcal{H}_p$. The initial state is $|L\rangle |0\rangle$. Each step of the walk consists of:\begin{enumerate}\item Measure the position of the system to see whether it is $1$. If the outcome is \textquotedblleft yes\textquotedblright, then the walk terminates, otherwise, it continues. The measurement is mathematically modelled by $M=\{M_\mathit{yes}, M_\mathit{no}\}$, where $$M_\mathit{yes}=|1\rangle\langle 1|,\ M_\mathit{no}=I_p-M_{yes}=\sum_{i\neq 1} |i\rangle\langle i|$$ and $I_p$ is the identity operator in the position space $\mathcal{H}_p$;  \item The Hadamard \textquotedblleft coin-tossing\textquotedblright\ operator $H$ is applied in the coin space $\mathcal{H}_c$; \item The shift operator $S$ defined by $$S|L,i\rangle=|L,i\ominus 1\rangle,\quad\quad S|R,i\rangle=|R,i\oplus 1\rangle$$ for $i=0,1,...,n-1$ is performed on the space $\mathcal{H}$. Intuitively, the system walks one step left or right according to the direction state. Here, $\oplus$ and $\ominus$ stand for addition and subtraction modulo $n$, respectively. The operator $S$ can be equivalently written as
$$S=\sum_{i=0}^{n-1}|L\rangle\langle L|\otimes|i\ominus 1\rangle\langle i|+\sum_{i=0}^{n-1} |R\rangle\langle R|\otimes|i\oplus 1\rangle\langle i|.$$
\end{enumerate}
This walk can be written as the quantum \textbf{while}-loop:
\begin{align*}\mathit{QW}\equiv\ \ c:=|L\rangle; p:=|0\rangle;&\mathbf{while}\ M[p]=\mathit{no}\ \mathbf{do}\ c:=H[c];\\ &c,p:=S[c,p]\ \mathbf{od}
\end{align*}

It is worth noticing an essential difference between the quantum walk and a classical random walk: the coin (or direction) variable $c$ can be in a superposition of $|L\rangle$ and $|R\rangle$ like $|+\rangle=\frac{1}{\sqrt{2}}(|L\rangle+|R\rangle)$, and thus the walker is moving left and right \textquotedblleft simultaneously\textquotedblright; for example, $$\frac{1}{\sqrt{2}}(|L\rangle+|R\rangle)|i\rangle\rightarrow \frac{1}{\sqrt{2}}(|L\rangle|i\ominus 1\rangle+|R\rangle|i\oplus 1\rangle).$$ This means that if the walker is currently at position $i$, then after one step she/he will be at both position $i\ominus 1$ and $i\oplus 1$.
\end{exam}

Similar to the case of classical \textbf{while}-programs, the operational semantics of quantum programs can be defined as a transition relation between configurations. For each quantum program $P$, we write $\mathit{var}(P)$ for the set of quantum variables occurring in $P$. Let $$\mathcal{H}_P=\bigotimes_{q\in\mathit{var}(P)}\mathcal{H}_q$$ be the state Hilbert space of $P$.  A partial density operator is defined as a positive operator $\rho$ with trace $\mathit{tr}(\rho)\leq 1$. We write $\mathcal{D}(\mathcal{H}_P)$ for the set of partial density operators in $\mathcal{H}_P$. A configuration is a pair $\langle P,\rho\rangle,$
where $P$ is a program or the termination symbol $\downarrow$, and $\rho\in\mathcal{D}(\mathcal{H}_P)$ denotes the state of quantum variables.

\begin{defn}[Operational Semantics \cite{Ying2011}, \cite{Ying2016}]\label{def-op-sem} The operational semantics of quantum programs is defined by the transition rules in Figure \ref{fig 3.1}. \begin{figure}\centering
\begin{equation*}\begin{split}&({\rm Sk})\ \ \langle\mathbf{skip},\rho\rangle\rightarrow\langle \downarrow,\rho\rangle\\ &({\rm In})\ \ \ \langle
q:=|0\rangle,\rho\rangle\rightarrow\langle \downarrow,\rho^{q}_0\rangle\\ 
&({\rm UT})\ \ \langle\overline{q}:=U[\overline{q}],\rho\rangle\rightarrow\langle
\downarrow,U\rho U^{\dag}\rangle\\ &({\rm SC})\ \ \ \frac{\langle P_1,\rho\rangle\rightarrow\langle
P_1^{\prime},\rho^{\prime}\rangle} {\langle
P_1;P_2,\rho\rangle\rightarrow\langle
P_1^{\prime};P_2,\rho^\prime\rangle}\\
&({\rm IF})\ \ \ \langle\mathbf{if}\ (\square m\cdot
M[\overline{q}]=m\rightarrow P_m)\ \mathbf{fi},\rho\rangle\rightarrow\langle
P_m,M_m\rho M_m^{\dag}\rangle\\
&({\rm L}0)\ \ \ \langle\mathbf{while}\
M[\overline{q}]=1\ \mathbf{do}\
P\ \mathbf{od},\rho\rangle\rightarrow\langle \downarrow, M_0\rho M_0^{\dag}\rangle\\
&({\rm L}1)\ \ \ \langle\mathbf{while}\
M[\overline{q}]=1\ \mathbf{do}\ P\ \mathbf{od},\rho\rangle\rightarrow\\ &\quad\quad\quad\quad\quad\quad\langle
P;\mathbf{while}\ M[\overline{q}]=1\ \mathbf{do}\ P\ \mathbf{od}, M_1\rho
M_1^{\dag}\rangle\end{split}\end{equation*}
\caption{Transition Rules for Quantum \textbf{while}-Programs.\ \ \ \ In (IN), $\rho^{q}_0=|0\rangle_q\langle 0|\rho|0\rangle_q\langle
0|+|0\rangle_q\langle 1|\rho|1\rangle_q\langle 0|$ if 
$\mathit{type}(q)=\mathbf{Bool}$ and $\rho^{q}_0=\sum_{n=-\infty}^{\infty}|0\rangle_q\langle n|\rho|n\rangle_q\langle
0|$ if $\mathit{type}(q)=\mathbf{Int}.$
In (SC), we make the convention $\downarrow;P_2=P_2.$
In (IF), $m$ ranges over every possible outcome of measurement $M=\{M_m\}.$}\label{fig 3.1}
\end{figure}\end{defn}

It is worth pointing out that the transition rules (In), (UT), (IF), (L0) and (L1) are decreed by the basic postulates of quantum mechanics. A major difference between the operational semantics of classical and quantum programs should be noticed. During the execution of a classical case statement or a \textbf{while}-loop, checking its guard does not change the state of program variables. In contrast, checking the guard of a quantum case statement or a quantum loop requires to perform a measurement on (the system denoted by) the program variables, which, again decreed by a basic postulate of quantum mechanics, changes the state of these variables. From the operational semantics, we see that the control flow of a quantum program is determined by the outcomes of the quantum measurements occurring in the case statements or loops. Note that the outcomes of quantum measurements are classical information, so the control flows of quantum programs considered in this paper are classical. Quantum programs with quantum control flows were studied in \cite{vanT}, \cite{AG05}, \cite{Zu}, \cite{Pan}, \cite{Val} and Chapters 6 and 7 of \cite{Ying2016}.     

\begin{defn}[Denotational Semantics \cite{Ying2011}, \cite{Ying2016}]\label{def-den-sem} For any program $P$, its semantic function is the mapping: \begin{align*}&\llbracket P\rrbracket:\mathcal{D}(\mathcal{H}_P)\rightarrow \mathcal{D}(\mathcal{H}_P)\\ &\llbracket P\rrbracket(\rho)=\sum\left\{|\rho^\prime: \langle P,\rho\rangle\rightarrow^\ast\langle \downarrow,\rho^\prime\rangle|\right\}\end{align*} for every $\rho\in\mathcal{D}(\mathcal{H}_P)$, where $\rightarrow^\ast$ is the reflexive and transitive closure of $\rightarrow$, and $\left\{|\cdot|\right\}$ denotes a multi-set.
\end{defn}

The defining equation of $\llbracket P\rrbracket (\rho)$ deserves an explanation. It is clear from rules (IF), (L0) and (L1) that transition relation $\rightarrow$ is usually not deterministic; that is, $\{|\rho^\prime:\langle P,\rho\rangle\rightarrow^\ast\langle\downarrow,\rho^\prime\rangle|\}$ may have more than one element. This is different from the case of classical \textbf{while}-programs. The summation in the right-hand side of the equation comes from the essential linearity of quantum mechanics. It is proved that $\llbracket P\rrbracket(\cdot)$ is a quantum operation (or a super-operator) (for the definition of quantum operation, see \cite{Nielsen}, Chapter 8 or \cite{Ying2016}, Chapter 2). 

To illustrate the above two definitions, let reconsider the three quantum dials in Example \ref{exam0}. 

\begin{exam}[Semantics of Three Quantum Dials] ---\ A Continuation of Example \ref{exam0}\label{exam2}. A state of probabilistic program $\mathit{flip}$ is a configuration of the slot machine, i.e.  a mapping from dials to suits. The semantics of $\mathit{flip}$ is a function that maps each initial state to a uniform distribution of states in which every configurations has probability $\frac{1}{8}$.

In contrast, the state Hilbert space of quantum program $\mathit{qflip}$ is then $\mathcal{H}_2^{\otimes 3}$. For instance, if we write: $$|+\rangle=\frac{1}{\sqrt{2}}(|\heartsuit\rangle+|\diamondsuit\rangle),\quad \quad |-\rangle=\frac{1}{\sqrt{2}}(|\heartsuit\rangle-|\diamondsuit\rangle)$$ for the equal superpositions of $|\heartsuit\rangle$ and $|\diamondsuit\rangle$, then: $$\llbracket \mathit{qflip}\rrbracket(|+,-,+\rangle)=|\heartsuit,\diamondsuit,\heartsuit\rangle;$$ if we write: $$|\mathit{W}\rangle=\frac{1}{\sqrt{3}}(|\heartsuit,\heartsuit,\diamondsuit\rangle+|\heartsuit,\diamondsuit,\heartsuit\rangle+|\diamondsuit,\heartsuit,\heartsuit\rangle)$$ for the Werner state, a typical entangled state of three qubits, then
\begin{align*}\llbracket \mathit{qflip}\rrbracket(|W\rangle)=\frac{1}{2\sqrt{6}}&(3|\heartsuit,\heartsuit,\heartsuit\rangle+|\heartsuit,\heartsuit,\diamondsuit\rangle+|\heartsuit,\diamondsuit,\heartsuit\rangle\\ &- \ |\heartsuit,\diamondsuit,\diamondsuit\rangle+|\diamondsuit,\heartsuit,\heartsuit\rangle-|\diamondsuit,\heartsuit,\diamondsuit\rangle\\ &-|\diamondsuit,\diamondsuit,\heartsuit\rangle-3|\diamondsuit,\diamondsuit,\diamondsuit\rangle).\end{align*} 

(Note that Definitions \ref{def-op-sem} and \ref{def-den-sem} are given in a more general way; i.e. in terms of density operators. Here, for simplicity, a pure state $|\psi\rangle$ is identified with the corresponding density operator $\rho=|\psi\rangle\langle\psi|$.)
\end{exam}

The L\"{o}wner order between operators in the state Hilbert space $\mathcal{H}_P$ of a quantum program is defined as follows: 
$$A\sqsubseteq B\ {\rm if}\ B-A\ {\rm is\ a\ positive\ operator.}$$ This order can be naturally lifted to an order between quantum operations (or super-operators). It can be shown that the set $\mathcal{D}(\mathcal{H}_P)$ of partial density operators in $\mathcal{H}_P$ with the L\"{o}wner order $\sqsubseteq$ is a complete partial order (CPO), and all quantum operations in $\mathcal{H}_P$ form a CPO too. Then a fixed point characterisation for the denotational semantics of quantum \textbf{while}-loops can be derived (see Proposition 5.1(6) and Corollary 5.1 in \cite{Ying2011}, or Proposition 3.3.2 and Corollary 3.3.1 in \cite{Ying2016}).   

\subsection{Partial Correctness and Total Correctness}

Now we review the notions of partial and total correctness for quantum programs. Recall from \cite{MM05} that probabilistic predicates (or assertions, e.g. preconditions and postconditions) are defined to be bounded real-valued functions interpreted as the expectations of random variables.  
As introduced in \cite{DP06}, a quantum predicate  in a Hilbert space $\mathcal{H}$ is an observable (a Hermitian operator) $A$ in $\mathcal{H}$ with $0\sqsubseteq A\sqsubseteq I$, where $0$ and $I$ are the zero operator and the identity operator in $\mathcal{H}$, respectively, and $\sqsubseteq$ stands for the L\"{o}wner order. In the quantum foundations literature, a quantum predicate is called an \textit{effect} (see, for example, \cite{Kr83}). 

\begin{defn}[Correctness Formula \cite{Ying2011}, \cite{Ying2016}] A correctness formula (or a Hoare triple) is a statement of the form $$\{A\}P\{B\}$$ where $P$ is a quantum program, and both $A, B$
are quantum predicates in $\mathcal{H}_P$, called the precondition and postcondition, respectively.\end{defn}

The appearance of a quantum Hoare triple is exactly the same as that of its classical counterpart. But in a classical Hoare triple, precondition $A$ and postcondition $B$ stands for (first-order) logical formulas, and in a quantum Hoare triple, $A$ and $B$ are two operators that are interpreted as observables in physics.  

\begin{defn}[Correctness \cite{Ying2011}, \cite{Ying2016}]\label{correctness-interpretation} --- Partial and Total. 
\begin{enumerate}\item The correctness formula $\{A\}P\{B\}$ is true in
the sense of total correctness, written $\models_{\mathit{tot}}\{A\}P\{B\},$ if for all input states 
$\rho\in\mathcal{D}(\mathcal{H}_P)$ we have: \begin{equation}\label{eq-total}\tr(A\rho)\leq \tr(B\llbracket P \rrbracket (\rho)).\end{equation}

\item The correctness formula $\{A\}P\{B\}$ is true in
the sense of partial correctness, written $\models_{\mathit{par}}\{A\}P\{B\},$ if for all input states 
$\rho\in\mathcal{D}(\mathcal{H}_P)$ we have: \begin{equation}\label{eq-partial}\tr(A\rho)\leq \tr(B\llbracket P \rrbracket (\rho))+
[\tr(\rho)-\tr(\llbracket P \rrbracket (\rho))].\end{equation}\end{enumerate}
\end{defn}

A key to understanding the above definition is the observation that $\mathit{tr}(A\rho)$ is the expected (average) value of observable $A$ when the system is in state $\rho$.  
The following lemma shows that the defining inequalities (\ref{eq-total}) and (\ref{eq-partial}) for total and partial correctness can be restated in a form of L\"{o}wner order, which is often easier to manipulate because the universal quantifier over density operator $\rho$ is eliminated.

Let us consider the three quantum dials in Examples \ref{exam0} and \ref{exam2} once again to illustrate the notion of correctness introduced in the above definition. 

\begin{exam}[Correctness of Three Quantum Dials] We write: $$|\mathit{GHZ}\rangle=\frac{1}{\sqrt{2}}(|\heartsuit,\heartsuit,\heartsuit\rangle+|\diamondsuit,\diamondsuit,\diamondsuit\rangle)$$ for the GHZ   (Greenberger-Horne-Zeilinger) state, another typical entangled state of three qubits, $$|\Phi\rangle=\frac{1}{2}(|\heartsuit,\heartsuit,\heartsuit\rangle+|\heartsuit,\diamondsuit,\diamondsuit\rangle+|\diamondsuit,\heartsuit,\diamondsuit\rangle+|\diamondsuit,\diamondsuit,\heartsuit\rangle),$$ and $|\Psi\rangle=|\heartsuit,\heartsuit,\heartsuit\rangle$. Then $A=|\Phi\rangle\langle\Phi|$, $B=|\Psi\rangle\langle\Psi|$ and $C=|\mathit{GHZ}\rangle\langle\mathit{GHZ}|$ are all quantum predicates. It is easy to check that $$\models_\mathit{tot}\{A\}\mathit{qflip}\{C\},\ \ \ \ \ \ \ \ \ \ \ \models_\mathit{tot}\{\frac{1}{4}B\}\mathit{qflip}\{C\}.$$ This means that if the input is state $|\Phi\rangle$, then program $\mathit{qflip}$ will certainly output the GHZ state; and if the input is state $|\Psi\rangle$, it will output a state $|\Gamma\rangle$ that is similar to the GHZ state in the sense: $$\Pr(|\Gamma\rangle\ {\rm and\ the\ GHZ\ state\ cannot\ be\ discriminated})\geq \frac{1}{4}.$$ 

Note that partial and total correctness are the same for $\mathit{qflip}$ because it does not contain any loop.
\end{exam}

The quantum predicates $A,B,C$ in the above example are very simple and defined by a particular input/output state. More examples of correctness specifications will be given in Section \ref{sec-examples}.

\begin{lem}\label{lem-correctness} \begin{enumerate}\item $\models_{\mathit{tot}}\{A\}P\{B\}$ if and only if $A\sqsubseteq\llbracket P\rrbracket^\ast(B)$. 
\item $\models_{\mathit{par}}\{A\}P\{B\}$ if and only if: $$A\sqsubseteq\llbracket P\rrbracket^\ast(B)+(I-\llbracket P\rrbracket^\ast(I))$$ where $I$ is the identity operator in $\mathcal{H}_P$, and $\llbracket P\rrbracket^\ast$ stands for the dual of super-operator $\llbracket P\rrbracket$ (see Definition \ref{def-dual} in Appendix \ref{app-a}). 
\end{enumerate}
\end{lem}

\begin{proof} Immediate from Lemma 4.1.2 in \cite{Ying2016} and Lemma \ref{lem-dual} in Appendix \ref{app-a}.
\end{proof}

The above lemma will be extensively used in the proofs of the results in Section \ref{auxiliary}.

\subsection{Quantum Hoare Logic}
A Hoare-like logic for quantum programs is established in \cite{Ying2011}. It consists of:
\begin{itemize}\item A proof system qPD for partial correctness: The axioms and inference rules of qPD are presented in Figure \ref{fig 3.2}.  
\begin{figure}\centering
\begin{equation*}\begin{split}
&({\rm Ax.Sk})\ \ \ \{A\}\mathbf{Skip}\{A\}\\ 
&({\rm Ax.In.B})\ \ \ \{|0\rangle_q\langle 0|A|0\rangle_q\langle 0|+|1\rangle_q\langle
0|A|0\rangle_q\langle 1|\}q:=|0\rangle\{A\}
 \\ &({\rm Ax.In.I}) \ \ \ \left\{\sum_{n=-\infty}^{\infty}|n\rangle_q\langle 0|A|0\rangle_q\langle
n|\right\}q:=|0\rangle\{A\}\\ 
&({\rm Ax.UT})\ \ \ 
\{U^{\dag}AU\}\overline{q}:=U\left[\overline{q}\right]\{A\}\\ 
&({\rm R.SC})\ \ \ 
\frac{\{A\}P_1\{B\}\ \ \ \ \ \ \{B\}P_2\{C\}}{\{A\}P_1;P_2\{C\}}\\
&({\rm R.IF})\ \ \ 
\frac{\{A_m\}P_m\{B\}\ {\rm for\ all}\ m}{\left\{\sum_m
M_m^{\dag}A_mM_m\right\}\mathbf{if}\ (\square m\cdot
M[\overline{q}]=m\rightarrow P_m)\ \mathbf{fi}\{B\}}\\ 
&({\rm R.LP})\ \ \ 
\frac{\{B\}P\left\{M_0^{\dag}AM_0+M_1^{\dag}BM_1\right\}}{\{M_0^{\dag}AM_0+M_1^{\dag}BM_1\}\mathbf{while}\
M[\overline{q}]=1\ \mathbf{do}\ P\ \mathbf{od}\{A\}}\\
&({\rm R.Or})\ \ \ \frac{A\sqsubseteq
A^{\prime}\ \ \ \ \{A^{\prime}\}P\{B^{\prime}\}\ \ \ \
B^{\prime}\sqsubseteq B}{\{A\}P\{B\}}
\end{split}\end{equation*}
\caption{Proof System qPD for Quantum \textbf{while}-Programs.\ \ \ \ In (Ax.In.B), $\mathit{type}(q)=\mathbf{Boolean}$. In (Ax.In.I), $\mathit{type}(q)=\mathbf{Int}$.}\label{fig 3.2}
\end{figure}
\item A proof system qTD for total correctness: qTD is obtained from qPD by replacing rule (R.LP) with the rule (R.LT) given in Figure \ref{fig 3.3}.
\begin{figure}\centering
\begin{equation*}\begin{split}
&({\rm R.LT})\ \ \ 
\frac{\begin{split}&\bullet\ \{Q\}S\{M_0^{\dag}PM_0+M_1^{\dag}QM_1\}\\ &\bullet\ {\rm
for\ any}\ \epsilon>0,\ t_\epsilon\ {\rm is\ a}\
(M_1^{\dag}QM_1,\epsilon){\rm -ranking}\\ &\quad\quad\quad{\rm function\ of\ loop}\
\mathbf{while}\ M[\overline{q}]=1\ \mathbf{do}\
S\ \mathbf{od}\end{split}}{\{M_0^{\dag}PM_0+M_1^{\dag}QM_1\}\mathbf{while}\
M[\overline{q}]=1\ \mathbf{do}\ S\ \mathbf{od}\{P\}}
\end{split}\end{equation*}
\caption{Proof System qTD for Quantum \textbf{while}-Programs.}\label{fig 3.3}
\end{figure} The notion of ranking function used in rule (R.LT) is defined in the following: 

\begin{defn}[\cite{Ying2011}, \cite{Ying2016}]\label{bou-def}Let $A$ be a quantum predicate in $\mathcal{H}_{\mathit{var}(P)\cup\overline{q}}$ and real number $\epsilon>0$. A function $$t: \mathcal{D}(\mathcal{H}_{\mathit{var}(P)\cup\overline{q}})\rightarrow\mathbb{N}\ ({\rm nonnegative\ integers})$$ is called a $(A,\epsilon)$-ranking 
function of quantum loop \textquotedblleft$\mathbf{while}\
M[\overline{q}]=1\ \mathbf{do}\ P\ \mathbf{od}$\textquotedblright\ if it
satisfies the following two conditions: for all $\rho\in
\mathcal{D}(\mathcal{H}_{\mathit{var}(P)\cup\overline{q}})$, \begin{enumerate}\item
$t\left(\llbracket S \rrbracket \left(M_1\rho M_1^{\dag}\right)\right)\leq t(\rho)$; and \item
$tr(P\rho)\geq\epsilon$ implies $t\left(\llbracket S \rrbracket \left(M_1\rho M_1^{\dag}\right)\right)<
t(\rho).$\end{enumerate} 
\end{defn} 
\end{itemize}

It is interesting to carefully compare the Hoare rule for partial correctness of classical loops:
 $$\frac{\{\varphi\wedge b\}P\{\varphi\}}{\{\varphi\}\mathbf{while}\
b\ \mathbf{do}\ P\ \mathbf{od}\{\varphi\wedge\neg b\}}$$ and its variant for total correctness with the corresponding rules (R.LP) and (R.LT) for quantum loops, respectively. At the first glance, the rules for classical loops and those for quantum loops are very different. But actually there is a certain similarity between them: in the rules for classical loops, $\varphi\wedge b$ (resp. $\varphi\wedge \neg b$) can be thought of as the restriction of $\varphi$ under $b$ (resp. $\neg b$). If we set: $$D=M_0^\dag AM_0 + M_1^\dag BM_1,$$ then $B$ can be seen as the restriction of $D$ under $M_1$ and $A$ the restriction of $D$ under $M_0$, and $M_0$ can be considered as the negation of $M_1$. Furthermore, as we will see at the end of Section \ref{invariant}, rules (R.LP) and (R.LT) can also be understood in the context of inductive assertion. 

The soundness and (relative) completeness of both proof system qPD and qTD were proved in \cite{Ying2011}, \cite{Ying2016}. 
\begin{thm}[Completeness \cite{Ying2011}, \cite{Ying2016}] For any quantum \textbf{while}-program $P$, and for any quantum predicates $A,B$,
\begin{align*}&\models_\mathit{par}\{A\}P\{B\}\Leftrightarrow\ \vdash_\mathit{qPD}\{A\}P\{B\},\\
&\models_\mathit{tot}\{A\}P\{B\}\Leftrightarrow\ \vdash_\mathit{qTD}\{A\}P\{B\}. 
\end{align*}
\end{thm}

\section{Examples of Quantum Correctness Specification}\label{sec-examples}

In this section, let us see several simple examples, which shows some interesting and tricky differences between specifying correctness of classical programs and that of quantum programs. 

\subsection{Unary Quantum Predicates} 

A quantum predicate about a single quantum system $q$; i.e. an observable (between the zero and identity operators) in state Hilbert space $\mathcal{H}_q$ can be understood as a unary quantum predicate. 

\begin{exam} An equality like $x=c$ with $x$ being a variable and $c$ a constant often appears in precondition or postcondition of a classical program correctness formula, e.g. $$\models_\mathit{tot}\{x=1\}x:=x+1\{x=2\}$$ Let $q$ be a quantum variable and $|\psi\rangle$ a given (constant) state in $\mathcal{H}_q$. Then equality \textquotedblleft $q=|\psi\rangle$\textquotedblright\ can be expressed as quantum predicate $A=|\psi\rangle\langle\psi|$ in $\mathcal{H}_q$, and we have: 
\begin{align}\label{u-1}&\models_\mathit{tot} \{U^\dag |\psi\rangle\langle\psi|U\}\ q:=U[q]\ \{A\},\\ &\label{u-2}\models_\mathit{tot} \{A\}\ q:=U[q]\ \{U|\psi\rangle\langle\psi|U^\dag\}\end{align} for any unitary operator $U$ in $\mathcal{H}_q$. The precondition in (\ref{u-1}) and the postcondition in (\ref{u-2}) express equalities \textquotedblleft $q=U^\dag|\psi\rangle$\textquotedblright, \textquotedblleft $q=U|\psi\rangle$\textquotedblright, respectively. More generally, for any closed subspace $X$ of $\mathcal{H}_q$, Membership \textquotedblleft $q\in X$\textquotedblright\ can be expressed as quantum predicate $P_X$ (the projection onto $X$) and we have: 
\begin{align}\label{x-1}&\models_\mathit{tot} \{P_{U^\dag (X)}\}\ q:=U[q]\ \{P_X\},\\ &\label{x-2}\models_\mathit{tot} \{P_X\}\ q:=U[q]\ \{P_{U(X)}\}\end{align}
The precondition in (\ref{x-1}) and the postcondition in (\ref{x-2}) express memberships \textquotedblleft $q\in U^\dag(X)$\textquotedblright, \textquotedblleft $q\in U(X)$\textquotedblright, respectively.
\end{exam}

\subsection{Quantum Relations}

A quantum predicate about multiple quantum systems $q_1,...,q_n$; i.e. an observable (between the zero and identity operators) in the tensor product: $$\bigotimes_{i=1}^n\mathcal{H}_{q_i},$$ can be understood as a quantum relation among $q_1,...,q_n$. In particular, a quantum predicate $A$ in $\mathcal{H}_1\otimes\mathcal{H}_2$ can be seen as a quantum (binary) relation between $\mathcal{H}_1$ and $\mathcal{H}_2$. The following are two examples showing some quantum relations used in preconditions and postconditions. A further discussion about quantum relation is given in Appendix \ref{app-relation}.

\begin{exam}\label{ex-symmetry} An equality like $x=y$ between two variables $x, y$ also often appears in precondition or postcondition of a classical program correctness formula, e.g. 
$$\models_\mathit{tot}\{x=y\}\ x:=x+1;y:=y+1\ \{x=y\}$$ \begin{enumerate}\item Let $p,q$ be two quantum variables with the same state Hilbert space $\mathcal{H}$. Then equality \textquotedblleft$p=q$\textquotedblright\ can be expressed as the symmetrisation operator: $$S_+=\frac{1}{2}(I+\mathrm{SWAP})$$ where $I$ is the identity operator in $\mathcal{H}\otimes\mathcal{H}$, and operator $\mathrm{SWAP}$ in $\mathcal{H}\otimes\mathcal{H}$ is defined by $$\mathrm{SWAP}|\varphi,\psi\rangle=|\psi,\varphi\rangle$$ for all $|\varphi\rangle,|\psi\rangle\in\mathcal{H}$, together with linearity.
It is easy to see that $\mathrm{SWAP}\cdot S_+=S_+\cdot \mathrm{SWAP}=\mathrm{SWAP}\cdot S_+\cdot \mathrm{SWAP}=S_+.$ Furthermore, for any unitary operator $U$ in $\mathcal{H}$, it holds that: 
$$\models_\mathit{tot}\{S_+\}\ p:=U[p]; q:=U[q]\ \{S_+\}$$ It is interesting to note that a similar conclusion holds for the anti-symmetrisation operator: $$S_-=\frac{1}{2}(I-\mathrm{SWAP})$$
that is, we have $\mathrm{SWAP}\cdot S_-=S_-\cdot \mathrm{SWAP}=-S_-,$ $\mathrm{SWAP}\cdot S_-\cdot \mathrm{SWAP}=S_-$ and $$\models_\mathit{tot}\{S_-\}\ p:=U[p]; q:=U[q]\ \{S_-\}$$ 
\item Equality \textquotedblleft $x=y$\textquotedblright\ can also be expressed in a different way; that is, as the operator $$=_\mathcal{B}\ =|\Psi\rangle\langle\Psi|$$ in $\mathcal{H}\otimes\mathcal{H}$, where $d=\dim\mathcal{H}$ and $|\Psi\rangle=\frac{1}{\sqrt{d}}\sum_i|ii\rangle$ is the maximally entangled state defined by an orthonormal basis $\mathcal{B}=\{|i\rangle\}$ of $\mathcal{H}$. This equality $=_\mathcal{B}$ has an intuitive explanation. The Hilbert-Schmidt inner product of two operators $A,B$ in $\mathcal{H}$: $$\langle A|B\rangle=\mathit{tr}\left(A^\dag B\right)$$ is often used to measure the similarity between $A$ and $B$. For two (mixed) states $\rho,\sigma$ in $\mathcal{H}$, it follows immediately from equation (9.73) in \cite{Nielsen} that $$\langle\rho |\sigma\rangle=d\cdot\langle\Psi|\rho\otimes\sigma|\Psi\rangle=d\cdot\mathit{tr}\left(=_\mathcal{B}(\rho\otimes\sigma)\right).$$ The quantity $\mathit{tr}\left(=_\mathcal{B}(\rho\otimes\sigma)\right)$ can be interpreted as the degree that $\rho, \sigma$ satisfies relation $=_\mathcal{B}$. Obviously, we have: 
$$\models_\mathit{tot}\{=_{U^\dag(\mathcal{B})}\}\ p:=U[p]; q:=U[q]\ \{=_\mathcal{B}\}$$ 
$$\models_\mathit{tot}\{=_\mathcal{B}\}\ p:=U[p]; q:=U[q]\ \{=_{U(\mathcal{B})}\}$$ 
where $=_{U(\mathcal{B})}$, $=_{U^\dag(\mathcal{B})}$ are the equalities defined by orthonormal bases $U(\mathcal{B})=\{U|i\rangle\}$ and $U^\dag(\mathcal{B})=\left\{U^\dag|i\rangle\right\}$, respectively. 
\end{enumerate}\end{exam}

\begin{exam} Consider the following correctness formula for classical programs: \begin{align*}\models_\mathit{tot} \{X=x\wedge Y=y\}\ &R:=X;X:=Y;Y:=R\\ &\{X=y\wedge Y=x\}\end{align*} It means that after executing the program, the values of variables $X$ and $Y$ are exchanged. To properly specify the precondition and postcondition, auxiliary (ghost) variables $x, y$ are introduced. In quantum computing, statements $R:=X, X:=Y$ and $Y:=R$ cannot be directly realised as prohibited by the no-cloning theorem \cite{WZ82}. But we can consider the following modification: let $X,Y, R$ be three quantum variables with the same state Hilbert space $\mathcal{H}$.  Then we have: \begin{equation}\label{no-ghost}\begin{split}\models_\mathit{tot}\ &\{A\}R,X:=\mathrm{SWAP}[R,X];Y,Y:=\mathrm{SWAP}[X,Y];\\ &Y,R:=\mathrm{SWAP}[Y,R]\ \{(\mathrm{SWAP}\otimes I)A(\mathrm{SWAP}\otimes I)\}\end{split}\end{equation} for any quantum predicate $A$ in $\mathcal{H}\otimes\mathcal{H}\otimes\mathcal{H}$, where the first, second and third $\mathcal{H}$ are the state Hilbert spaces of $X,Y,R$ respectively, and $I$ is the identity operator in $\mathcal{H}$. Note that we do not use any auxiliary (ghost) variable in correctness formula (\ref{no-ghost}). In particular, if $$A=|\varphi\rangle\langle\varphi|\otimes |\psi\rangle\langle\psi|\otimes I,$$ where $|\varphi\rangle,|\psi\rangle\in\mathcal{H}$, and $I$ is the identity operator in $\mathcal{H}$, then it holds that 
\begin{equation}\label{ghost}\begin{split}\models_\mathit{tot}\ &\{|\varphi\rangle\langle\varphi|\otimes |\psi\rangle\langle\psi|\otimes I\}R,X:=\mathrm{SWAP}[R,X];\\ &Y,Y:=\mathrm{SWAP}[X,Y]; Y,R:=\mathrm{SWAP}[Y,R]\\ &\{|\psi\rangle\langle\psi|\otimes |\varphi\rangle\langle\varphi|\otimes I\}\end{split}\end{equation} Intuitively, the precondition and postcondition of (\ref{ghost}) means \textquotedblleft $X=|\varphi\rangle$ and $Y=|\psi\rangle$\textquotedblright, \textquotedblleft $X=|\psi\rangle$ and $Y=|\varphi\rangle$\textquotedblright, respectively. Here, auxiliary variables $|\varphi\rangle, |\psi\rangle$ are employed. 
\end{exam}

All of the examples given in the previous subsection and this subsection can be proved using axiom (Ax.UT) and rule (R.SC) in Figure \ref{fig 3.2}. A more interesting example using rule (R.IF) will be presented in Section \ref{outline}.

\subsection{Logical Connectives}

Logical connectives are widely used in specifying (the preconditions and postconditions of) classical programs. Here, we show that several quantum counterparts of them can be used in specifying correctness of quantum programs. However, it is not the case that every logical connective has an appropriate quantum counterpart.  

\begin{exam} The predicate \textquotedblleft True\textquotedblright\ and \textquotedblleft False\textquotedblright\ are described by the identity and zero operators $I$ and $0$, respectively, in the state Hilbert space $\mathcal{H}_P$ of the program $P$ under consideration.  \begin{itemize}\item $\models_\mathit{par}\{I\}P\{B\}$ means that for any density operator $\rho$, $$\mathit{tr}(B\llbracket P\rrbracket(\rho))=\mathit{tr}(\llbracket P\rrbracket(\rho));$$ that is, the probability that postcondition $B$ is satisfied by the output is equal to the probability that the program terminates. 
\item $\models_\mathit{par}\{A\}P\{I\}$ always holds, because $A\sqsubseteq I$ implies that for any partial density operator $\rho$, $$\mathit{tr}(A\rho)\leq\mathit{tr}(\rho)=\mathit{tr}(\llbracket P\rrbracket(\rho))+[\mathit{tr}(\rho)-\mathit{tr}(\llbracket P\rrbracket(\rho))].$$
\item $\models_\mathit{tot}\{I\}P\{B\}$ means that for any density operator $\rho$, $$\mathit{tr}(B\llbracket P\rrbracket(\rho))=1.$$ Since $B\sqsubseteq I$, we have $\mathit{tr}(B\llbracket P\rrbracket(\rho))\leq\mathit{tr}(\llbracket P\rrbracket(\rho))$ and thus $\mathit{tr}(\llbracket P\rrbracket(\rho))=1$; that is, program $P$ always terminates, and output $\llbracket P\rrbracket(\rho)$ satisfies postcondition $B$. 
\item $\models_\mathit{tot}\{A\}P\{I\}$ means that for any density operator, $\mathit{tr}(A\rho)\leq\mathit{tr}(\llbracket P\rrbracket(\rho))$; intuitively, if precondition $A$ is satisfied by input $\rho$ to $P$, then $P$ terminates.  
\end{itemize}
\end{exam}

\begin{exam} \begin{enumerate}
\item Negation in classical logic has a quantum counterpart: for each quantum predicate in Hilbert space $\mathcal{H}$, $I-A$ can be used as the negation of $A$. In fact, for any density operator $\rho$ in $\mathcal{H}$, $$\mathit{tr}[(I-A)\rho]=1-\mathit{tr}(A\rho).$$ 

\item Conjunction in classical logic does not always have an appropriate quantum counterpart; some difficulties caused by this fact in reasoning about quantum programs were discussed in \cite{Ying07}. But if $A_1$ and $A_2$ are quantum predicates in $\mathcal{H}_{1}$, $\mathcal{H}_{2}$, respectively, then $A_1\otimes A_2$ can be used as the conjunction (in $\mathcal{H}_1\otimes\mathcal{H}_2$) of $A_1$ and $A_2$ because for any density operator $\rho_1$ in $\mathcal{H}_{1}$ and $\rho_2$ in $\mathcal{H}_{2}$, it holds that $$\mathit{tr}((A_1\otimes A_2)(\rho_1\otimes\rho_2))=\mathit{tr}(A_1\rho_1)\cdot\mathit{tr}(A_2\rho_2);$$ for example, see equation (\ref{ghost}).
\item For a family $\{A_i\}$of quantum predicates in the same Hilbert space $\mathcal{H}$, and a probability distribution $\{p_i\}$, we often use the convex combination $\sum_ip_iA_i$ as a connective; for example, see rules (R.CC) and (R.Inv) in Figure \ref{fig 3.4}. 
\item In a sense, quantum predicate $\sum_m M_m^\dag A_mM_m$ in rule (R.IF) in Figure \ref{fig 3.2} can be understood as the disjunction of (the conjunction of) \textquotedblleft the outcome of measurement $M$ is $m$\textquotedblright\ and quantum predicate $A_m$. 
\end{enumerate}
\end{exam}

\section{Proof Outlines}\label{outline}

The notion of proof outline can be introduced to structure a correctness proof of a classical \textbf{while}-program, according to the structure of the program itself, so that the proof is easier to follow (see \cite{Apt09}, Section 3.4). 
It is also a basis for defining non-interference in reasoning about parallel programs with shared variables.  
In this section, we generalise this notion to the quantum case. 

\begin{defn}\label{def-outline} Let $S$ be a quantum \textbf{while}-programs. \begin{enumerate}\item A proof outline for partial correctness of $S$ is a formula $\{A\}P^\ast\{B\}$ formed by the formation axioms and rules in Figure \ref{fig 3.2+}, where $P^\ast$ results from $P$ by interspersing quantum predicates. 
\begin{figure*}[!t]
\normalsize
\begin{equation*}\begin{split}
&({\rm Ax.Sk'})\ \ \ \{A\}\mathbf{Skip}\{A\}\\ 
&({\rm Ax.In.B'})\ \ \ \{|0\rangle_q\langle 0|A|0\rangle_q\langle 0|+|1\rangle_q\langle
0|A|0\rangle_q\langle 1|\}q:=|0\rangle\{A\}
 \\ &({\rm Ax.In.I'}) \ \ \ \left\{\sum_{n=-\infty}^{\infty}|n\rangle_q\langle 0|A|0\rangle_q\langle
n|\right\}q:=|0\rangle\{A\}\\ 
&({\rm Ax.UT'})\ \ \ 
\{U^{\dag}AU\}\overline{q}:=U\left[\overline{q}\right]\{A\}\\ 
&({\rm R.SC'})\ \ \ 
\frac{\{A\}P^\ast_1\{B\}\ \ \ \ \ \ \{B\}P^\ast_2\{C\}}{\{A\}P^\ast_1;\{B\}P^\ast_2\{C\}}\\
&({\rm R.IF'})\ \ \ 
\frac{\{A_{m_i}\}P^\ast_{m_i}\{B\}\ (i=1,...,k)}{\begin{array}{ccc}\left\{\sum_{i=1}^{k}
M_{m_i}^{\dag}A_{m_i}M_{m_i}\right\}\ \mathbf{if}\  
M[\overline{q}]=m_1\rightarrow \left\{A_{m_1}\right\}P^\ast_{m_1}\\ 
\ \ \ \ \ \ \ \ \ \ \ \ \ \ \ \ \ \ \ \ \ \ \ \ \ \ \ \ \ .................. \\ \\ 
\ \ \ \ \ \ \ \ \ \ \ \ \ \ \ \ \ \ \ \ \ \ \ \ \ \ \ \ \ \ \ \ \ \ \ \ \ \ \ \ \ \ \ \square\ M[\overline{q}]=m_k\rightarrow \left\{ A_{m_k}\right\}P^\ast_{m_k}\\ \ \ \ \ \mathbf{fi}\ \{B\}\end{array}}\\ 
&({\rm R.LP'})\ \ \ 
\frac{\{B\}P^\ast\left\{M_0^{\dag}AM_0+M_1^{\dag}BM_1\right\}}{\begin{array}{ccc}\{M_0^{\dag}AM_0+M_1^{\dag}BM_1\}\ \mathbf{while}\
M[\overline{q}]=1\ \mathbf{do}\ \left\{ B\right\}\ P^\ast\ \left\{M_0^{\dag}AM_0+M_1^{\dag}BM_1\right\}\ \mathbf{od}\ \{A\}\end{array}}\\
&({\rm R.Or'})\ \ \ \frac{A\sqsubseteq
A^{\prime}\ \ \ \ \{A^{\prime}\}P^\ast\{B^{\prime}\}\ \ \ \
B^{\prime}\sqsubseteq B}{\{A\}\{A^\prime\}P\{B^\prime\}\{B\}}\\
&({\rm R.Del})\ \ \ \ \frac{\{A\}P^\ast\{B\}}{\{A\}P^{\ast\ast}\{B\}}
\end{split}\end{equation*}
\caption{Formation Axioms and Rules for Partial Correctness of Quantum \textbf{while}-Programs.\ \ \ \ In (R.IF'), $\{m_1,...,m_k\}$ is the set of all possible outcomes of measurement $M$. In (R.Del), $S^{\ast\ast}$ is obtained by deleting some quantum predicates from $S^\ast$.}\label{fig 3.2+}
\hrulefill \vspace*{4pt}
\end{figure*}

\item A proof outline for total correctness of $S$ is defined by introducing ranking function into rule (R.LP') (see Definition \ref{bou-def} and rule (R.LT)).  
\end{enumerate}
\end{defn}

Let us consider a simple example to illustrate the above definition. 

\begin{exam}[Teleportation \cite{Benn92}] Quantum teleportation is a protocol that can send quantum states only using a classical communication channel. Suppose that Alice possesses two qubits $p,q$ and Bob possesses qubit $r$, and there is entanglement, i.e. the EPR (Einstein-Podolsky-Rosen) pair: $$|\beta_{00}\rangle=\frac{1}{\sqrt{2}}(|00\rangle+|11\rangle),$$ between $q$ and $r$. Then Alice can send a quantum state $|\psi\rangle=\alpha|0\rangle+\beta|1\rangle$ to Bob, i.e. from $p$ to $r$, by two-bit classical communication (for detailed description, see \cite{Nielsen}, Section 1.3.7). This protocol can be written as quantum program QTEL in Figure \ref{fig 3.tel}. The correctness of QTEL can be described as the Hoare triple: $$\left\{|\psi\rangle_p\langle\psi|\otimes|\beta_{00}\rangle_{q,r}\langle\beta_{00}|\right\}{\rm QTEL} \left\{I_p\otimes I_q\otimes |\psi\rangle_r\langle\psi|\right\}
.$$ A proof outline for the correctness of QTEL is also presented in Figure \ref{fig 3.tel}. 
\begin{figure*}[!t]
\normalsize
\begin{align*}&\left\{|\psi\rangle_p\langle\psi|\otimes|\beta_{00}\rangle_{q,r}\langle\beta_{00}|\right\}\\
&\sqsubseteq\{|\beta_{00}\rangle_{p,q}\langle\beta_{00}|\otimes |\psi\rangle_r\langle\psi|+|\beta_{10}\rangle_{p,q}\langle\beta_{10}|\otimes |\psi_1\rangle_r\langle\psi_1|\\ & \ \ \ \ \ \ \ \ \ \ \ \ \ \ \ \ \ \ \ \ \ \ \ \ \ \ \ \ \ \ \ \ \ \ \ \ \ \ \ \ \ \ \ \ \ \ +|\beta_{01}\rangle_{p,q}\langle\beta_{01}|\otimes |\psi_2\rangle_r\langle\psi_2|+|\beta_{11}\rangle_{p,q}\langle\beta_{11}|\otimes |\psi_3\rangle_r\langle\psi_3|\}\\ 
{\rm QTEL}\equiv\ &p,q:=\ {\rm CNOT}[p,q];\\
&\{|+\rangle_p\langle +|\otimes |0\rangle_q\langle 0|\otimes |\psi\rangle_r\langle\psi|+|-\rangle_p\langle -|\otimes |0\rangle_q\langle 0|\otimes |\psi_1\rangle_r\langle\psi_1|\\ &\ \ \ \ \ \ \ \ \ \ \ \ \ \ \ \ \ \ \ \ \ \ \ \ \ \ \ \ \ \ \ \ \ \ \ \ \ \ \ \ \ \ \ \ \ \ \ \ \ +|+\rangle_p\langle +|\otimes |1\rangle_q\langle 1|\otimes |\psi_2\rangle_r\langle\psi_2|+|-\rangle_p\langle -|\otimes |1\rangle_q\langle 1|\otimes |\psi_3\rangle_r\langle\psi_3|\}\\ 
&p:=H[p];\\ 
&\{|0\rangle_p\langle 0|\otimes |0\rangle_q\langle 0|\otimes |\psi\rangle_r\langle\psi|+|1\rangle_p\langle 1|\otimes |0\rangle_q\langle 0|\otimes |\psi_1\rangle_r\langle\psi_1|\\ & \ \ \ \ \ \ \ \ \ \ \ \ \ \ \ \ \ \ \ \ \ \ \ \ \ \ \ \ \ \ \ \ \ \ \ \ \ \ \ \ \ \ \ \ \ \ \ +|0\rangle_p\langle 0|\otimes |1\rangle_q\langle 1|\otimes |\psi_2\rangle_r\langle\psi_2|+|1\rangle_p\langle 1|\otimes |1\rangle_q\langle 1|\otimes |\psi_3\rangle_r\langle\psi_3|\}
\\ &\mathbf{if}\ M[q]=0\rightarrow \left\{|0\rangle_p\langle 0|\otimes I_q\otimes |\psi\rangle_r\langle\psi|+|1\rangle_p\langle 1|\otimes I_q\otimes |\psi_1\rangle_r\langle\psi_1|\right\} \mathbf{skip}\\
&\square\ \ \ \ \ \ \ \ \ \ \ \ \ \ \ 1\rightarrow \left\{|0\rangle_p\langle 0|\otimes I_q\otimes |\psi_2\rangle_r\langle\psi_2|+|1\rangle_p\langle 1|\otimes I_q\otimes |\psi_3\rangle_r\langle\psi_3|\right\} r:=X[r]\\
&\mathbf{fi};\\
&\left\{|0\rangle_p\langle 0|\otimes I_q\otimes |\psi\rangle_r\langle\psi|+|1\rangle_p\langle 1|\otimes I_q\otimes |\psi_1\rangle_r\langle\psi_1|\right\}\\
&\mathbf{if}\ M[p]=0\rightarrow \left\{I_p\otimes I_q\otimes |\psi\rangle_r\langle\psi|\right\} \mathbf{skip}\\
&\square\ \ \ \ \ \ \ \ \ \ \ \ \ \ \ 1\rightarrow \left\{I_p\otimes I_q\otimes |\psi_1\rangle_r\langle\psi_1|\right\} r:=Z[r]\\
&\mathbf{fi}\left\{I_p\otimes I_q\otimes |\psi\rangle_r\langle\psi|\right\}
\end{align*}
\caption{Quantum Teleportation Program and Correctness Proof Outline. \ \ \ \ (1) Quantum states: $|\psi_1\rangle=\alpha|0\rangle-\beta|1\rangle$, $|\psi_2\rangle=\beta|0\rangle+\alpha|1\rangle$, $|\psi_3\rangle=-\beta|0\rangle+\alpha|1\rangle$. (2) Entangled states: $|\beta_{01}=\frac{1}{\sqrt{2}}(|01\rangle+|10\rangle$, $|\beta_{10}=\frac{1}{\sqrt{2}}(|00\rangle-|11\rangle$, $|\beta_{01}=\frac{1}{\sqrt{2}}(|01\rangle-|10\rangle$. (3) Measurement in the computational basis: $M=\{M_0,M_1\}$, where $M_0=|0\rangle\langle 0|$, $M_1=|1\rangle\langle 1|$.}\label{fig 3.tel}
\hrulefill \vspace*{4pt}
\end{figure*}
\end{exam}

We can define the notion of subprogram of a quantum program in a familiar way; namely by induction on the length of the program. Then a special form of proof outline is defined in the following: 

\begin{defn}A proof outline $\{A\}P^\ast\{B\}$ of $P$ is called standard if every subprogram $Q$ of $P$ is proceeded by exactly one quantum predicate, denoted $\mathit{pre}(Q)$, in $P^\ast$.\end{defn}

The next proposition shows that standard proof outlines are indeed sufficient for reasoning about correctness of quantum programs. 

\begin{prop}\begin{enumerate}\item If $\{A\}P^\ast\{B\}$ is a proof outline for partial correctness (respectively, total correctness), then $\vdash_\mathit{qPD}\{A\}P\{B\}$ (respectively, $\vdash_\mathit{qTD}\{A\}P\{B\}$). 
\item If $\vdash_\mathit{qPD}\{A\}P\{B\}$ (respectively, $\vdash_\mathit{qTD}\{A\}P\{B\}$), then there is a standard proof outline $\{A\}P^\ast\{B\}$ for partial correctness (respectively, total correctness).
\end{enumerate}
\end{prop}

\begin{proof}Induction on the lengths of proof and formation.\end{proof}

As in classical programming, another usage of proof outline is that it enables us to establish a strong soundness of quantum Hoare logic. To this end, we need the following:  

\begin{defn}\label{def-head} Let $P$ be a quantum \textbf{while}-program and $T$ a subprogram of $P$. Then $\mathit{at}(T,P)$ is inductively defined as follows:\begin{enumerate}\item If $T\equiv P$, then $\mathit{at}(T,P)\equiv P$;
\item If $P\equiv P_1;P_2$, then $\mathit{at}(T,P)\equiv\mathit{at}(T,P_1);P_2$ when $T$ is a subprogram of $P_1$, and $\mathit{at}(T,P)\equiv\mathit{at}(T,P_2)$ when $T$ is a subprogram of $P_2$; \item If $P\equiv\mathbf{if}\ (\square m\cdot M[\overline{q}]=m\rightarrow P_m)\ \mathbf{fi}$, then for each $m$, whenever $T$ is a subprogram of $P_m$, $\mathit{at}(T,P)\equiv\mathit{at}(T,P_m)$; \item If $P\equiv\mathbf{while}\ M[\overline{q}]=1\ \mathbf{do}\ P^\prime\ \mathbf{od}$ and $T$ is a subprogram of $P^\prime$, then $\mathit{at}(T,P)\equiv\mathit{at}(T,P^\prime);P$.
\end{enumerate}
\end{defn}

Intuitively, $\mathit{at}(T,P)$ denotes the remainder of $P$ to be executed when the control is at subprogram $T$. 

Definitions \ref{def-outline} to \ref{def-head} are straightforward generalisations of the corresponding concepts for classical programs (see for example \cite{Apt09}, Section 3.4). However, the strong soundness (Theorem 3.3 in \cite{Apt09}) cannot be straightforwardly generalised to the quantum case. To present the strong soundness for quantum programs, we also need to extend the transition relation between configurations given in Definition \ref{def-op-sem} to a transition relation between configuration ensembles. We define a configuration ensemble as a multi-set $\{|\langle P_i,\rho_i\rangle|\}$ of configurations with $$\sum_i\mathit{tr}(\rho_i)\leq 1.$$ 
For simplicity, we identify a singleton $\{|\langle P,\rho\rangle|\}$ with the configuration $\langle P,\rho\rangle$.  

\begin{defn} The transition relation between configuration ensembles is of the form: 
\begin{equation}\label{op-sem-extend}\{|\langle P_i,\rho_i\rangle|\}\rightarrow\{|\langle Q_j,\sigma_j\rangle|\}\end{equation} and defined by rules (Sk), (In), (UT), (SC) in Figure \ref{fig 3.1} together with the rules presented in Figure \ref{fig ext-3.1}.
\begin{figure*}[!t]\normalsize
\begin{equation*}\begin{split}
&({\rm IF'})\ \ \ \langle\mathbf{if}\ (\square m\cdot
M[\overline{q}]=m\rightarrow P_m)\ \mathbf{fi},\rho\rangle\rightarrow\{|\langle
P_m,M_m\rho M_m^{\dag}\rangle|\}\\
&({\rm L}')\ \ \ \langle\mathbf{while}\
M[\overline{q}]=1\ \mathbf{do}\
P\ \mathbf{od},\rho\rangle\rightarrow \{|\langle \downarrow, M_0\rho M_0^{\dag}\rangle, \langle
P;\mathbf{while}\ M[\overline{q}]=1\ \mathbf{do}\ P\ \mathbf{od}, M_1\rho
M_1^{\dag}\rangle|\}\\
&({\rm MS})\ \ \ \frac{\langle P,\rho\rangle\in\mathcal{A}\ \ \ \ \langle P,\rho\rangle\rightarrow\mathcal{B}}{\mathcal{A}\rightarrow\left(\mathcal{A}\setminus\{|\langle P,\rho\rangle|\}\right)\cup\mathcal{B}}
\end{split}\end{equation*}
\caption{Extended Transition Rules for Quantum \textbf{while}-Programs.\ \ \ \ In (MU), $\mathcal{A}, \mathcal{B}$ are configuration ensembles.}\label{fig ext-3.1}
\hrulefill \vspace*{4pt}
\end{figure*}
\end{defn}

Now we are ready to present the strong soundness of the quantum Hoare logic. 

\begin{thm}[Strong Soundness]\label{thm.strong-sound} Let $\{A\}P^\ast\{B\}$ be a standard proof outline for partial correctness. If $$\langle P,\rho\rangle\rightarrow^\ast\{|\langle P_i,\rho_i\rangle|\},$$ then: \begin{enumerate} 
\item for each $i$, $P_i\equiv \mathit{at}(T_i,P)$ for some subprogram $T_i$ of $P$ or $P_i\equiv\ \downarrow$; and 
\item $\mathit{tr}(A\rho)\leq\sum_i\mathit{tr}\left(B_i\rho_i\right)$, where $$B_i=\begin{cases}B\ &{\rm if}\ P_i\equiv\ \downarrow,\\
\mathit{pre}\left(T_i\right)\ &{\rm if}\ P_i\equiv\mathit{at}\left(T_i,P\right).
\end{cases}$$
\end{enumerate}
\end{thm}

\begin{proof}See Appendix \ref{proof.strong-sound}.\end{proof}

A difference between the above theorem and the strong soundness for classical \textbf{while}-programs (see \cite{Apt09}, Theorem 3.3) should be noticed. The summation in the right-hand side of the inequality in clause 2 of the above theorem indicates that we have to consider all the reached programs $P_i$ collectively in the quantum case. The reason behind it is that certain nondeterminism is introduced in the operational semantics of case statements and \textbf{while}-loops in quantum computation.  

\section{Auxiliary Axioms and Rules}\label{auxiliary}

Several auxiliary axioms and rules introduced in \cite{Gor75}, \cite{Harel79} (see also \cite{Apt09}, Section 3.8) are very useful for simplifying the presentation of correctness proofs of classical programs. In this section, we generalise some of them into the quantum case. For this purpose, we first introduce several notations. Let $X\subseteq Y\subseteq\mathit{Var}$, and let $A$ be an operator in $$\mathcal{H}_X=\bigotimes_{q\in X}\mathcal{H}_q.$$ Then operator: $$\mathit{cl}_Y(A)=A\otimes I_{\mathcal{H}_{Y\setminus X}}$$ is called the cylindric extension of $A$ in $\mathcal{H}_Y$. If $X,Y\subseteq\mathit{Var}$ and $X\cap Y=\emptyset$. Then partial trace $\mathit{tr}_Y$ is a mapping from operators in $\mathcal{H}_{X\cup Y}$ to operators in $\mathcal{H}_X$ defined by $$\mathit{tr}_Y(|\varphi\rangle\langle\psi|\otimes |\varphi^\prime\rangle\langle\psi^\prime|)=\langle\psi^\prime|\varphi^\prime\rangle\cdot |\varphi\rangle\langle\psi|$$ for every $|\varphi\rangle,|\psi\rangle$ in $\mathcal{H}_X$ and $|\varphi^\prime\rangle,|\psi^\prime\rangle$ in $\mathcal{H}_Y$. 

Now we can present the auxiliary axioms and rules for quantum programs in Figure \ref{fig 3.4}.

\begin{figure*}[!t]\normalsize
\begin{align*}
&({\rm Ax.Inv})\ \ \ \left\{A\right\}P\left\{A\right\}\\
&({\rm R.TI})\ \ \ \frac{\{A\}P\{B\}}{\left\{\mathit{tr}_WA\right\}P\left\{B\right\}}\\
&({\rm R.CC})\ \ \ \frac{\left\{A_i\right\}P\left\{B_i\right\}\ (i=1,...,m)}{\left\{\sum_{i=1}^mp_iA_i\right\}P\left\{\sum_{i=1}^mp_iB_i\right\}}\\
&({\rm R.Inv})\ \ \ \frac{\{A\}P\{B\}}{\{pA+qC\}P\{pB+qC\}}\\ 
&({\rm R.SO})\ \ \ \frac{\{A\}P\{B\}}{\left\{\mathcal{E}^\ast(A)\right\}P\left\{\mathcal{E}^\ast(B)\right\}}\\
&({\rm R.Lim})\ \ \ \frac{\lim_{n\rightarrow\infty}A_n=A\ \ \ \left\{A_n\right\}P\left\{B_n\right\}\ \ \ \lim_{n\rightarrow\infty}B_n=B}{\{A\}P\{B\}}
\end{align*}
\caption{Auxiliary Axioms and Rules.\ \ \ \ In (Ax.Inv), $\mathit{var}(P)\cap V=\emptyset$ and $A=\mathit{cl}_{V\cup\mathit{var}(P)}(B)$ for some $V\subseteq\mathit{Var}$ and for some quantum predicate $B$ in $\mathcal{H}_V.$ In (R.TI), $V,W\subseteq \mathit{Var},$ $V\cap W=\emptyset,$ $A, B$ are quantum predicates in $\mathcal{H}_{V\cup W}$ and $\mathcal{H}_V,$ respectively, and $\mathit{var}(P)\subseteq V.$ In  (R.CC), $p_i\geq 0$ $(i=1,...,m)$ and $\sum_{i=1}^m p_j\leq 1.$ 
In (R.Inv), $p,q\geq 0$, $p+q\leq 1$, and $C$ is a quantum predicate in $\mathcal{H}_V$ for some $V\subseteq\mathit{Var}$ with $V\cap\mathit{var}(P)=\emptyset$.
In (R.SO), $\mathcal{E}$ is a quantum operation (or super-operator) in $\mathcal{H}_V$ for some $V\subseteq\mathit{Var}$ with $V\cap\mathit{var}(P)=\emptyset$. In (R.Lim), $\{A_n\}$ and $\{B_n\}$ are increasing and decreasing, respectively sequences with respect to the L\"{o}wner order.}\label{fig 3.4}
\hrulefill \vspace*{4pt}
\end{figure*}

It is interesting to compare our auxiliary axioms and rules for quantum programs with those for classical programs: 
\begin{itemize}
\item The appearance of axiom (Ax.Inv) is the same as axiom (INVARIANCE) in Section 3.8 of \cite{Apt09}. 
\item Obviously, rule (R.SO) is quantum generalisations of rule (SUBSTITUTION) in \cite{Apt09}, with the substitution $\overline{z}:=\overline{t}$ being replaced by an arbitrary super-operator $\mathcal{E}$. 
\item It is easy to see that rules (R.CC) and (R.Inv) are generalisations of rules (CONJUNCTION) and (INVARIANCE), respectively, in \cite{Apt09}, with logical conjunction being replaced by a probabilistic (convex) combination. 
\item It is more interesting to note that rule (R.TI) is a quantum generalisation of two rules (DISJUNCTION) and ($\exists$-INTRODUCTION) in \cite{Apt09}, where partial trace is considered as a quantum counterpart of logical disjunction and existence quantifier.  
\end{itemize}

The following theorem establishes soundness of the auxiliary axioms and rules in Figure \ref{fig 3.4}. 

\begin{thm}\label{Aux-Sound}\begin{enumerate}\item The axiom (Ax.Inv) is sound for partial correctness.
\item The rules (R.TI), (R.CC), (R.Inv), (R.SO) and (R.Lim) are sound both for partial and total correctness. 
\end{enumerate}\end{thm}

\begin{proof} See Appendix \ref{proof.1}
\end{proof}

Before concluding this section, let us see a simple example to show how the auxiliary axioms and rules can be used to derive some new properties of a quantum program. 

\begin{exam} Let $q$ be a qubit variable. Then by (Ax.UT) in Figure \ref{fig 3.3} we obtain: 
\begin{equation}\label{ax-Had}\{|-\rangle\langle -|\}\ q:=H[q]\ \{|1\rangle\langle 1|\}\end{equation} where $H$ is the Hadamard gate and $|-\rangle=\frac{1}{\sqrt{2}}(|0\rangle-|1\rangle)$. 
The rule (R.SO) in Figure \ref{fig 3.4} is especially useful for reasoning about quantum computation with noise. The amplitude damping channel is an important type of quantum noise where energy is lost from a quantum system with a certain probability $\gamma$. It  can be modelled by quantum operation:  
$$\mathcal{E}(\rho)=E_0\rho E_0+E_1\rho E_1$$ 
 for any density operator $\rho$, where 
$$E_0=\left(\begin{array}{cc} 1 & 0\\ 0 &\sqrt{1-\gamma}
\end{array}\right),\quad\quad E_1=\left(\begin{array}{cc} 1 & \sqrt{\gamma}\\ 0 &0
\end{array}\right).$$ 
Then applying (R.SO) to (\ref{ax-Had}) yields: 
\begin{align*}&\left\{\frac{1+\gamma}{2}|0\rangle\langle 0|-\frac{\sqrt{1-\gamma}}{2}(|0\rangle\langle 1|+|1\rangle\langle 0|)+\frac{1-\gamma}{2}|1\rangle\langle 1|\right\}\\ &\quad\quad\quad\quad q:=H[q]\ \left\{\gamma |0\rangle\langle 0|+(1-\gamma)|1\rangle\langle 1|\right\}.\end{align*}
 \end{exam}

\section{Mechanising Quantum Program Verification}\label{mechanising}

Proving correctness of classical programs using Hoare logic is very fiddly and tedious. As we saw from the example given in Figure \ref{fig 3.tel}, it is even worse to prove correctness of quantum programs because a huge amount of vector and matrix calculations is involved. Thus, automatic tools will be very helpful. 

Similar to the case of classical programs, a verifier can prove the correctness: $$\vdash_\mathit{par}\{A\}P\{B\}\ ({\rm or}\ \vdash_\mathit{tot}\{A\}P\{B\})$$ of a quantum program $P$ in the following three steps:
\begin{itemize}\item \textit{Annotating the program}: insert quantum predicates at the intermediate points (in a way similar to a proof outline described in Section \ref{outline}). Intuitively, an inserted statement is intended to hold when the control reaches the corresponding point.  
\item \textit{Generating verification conditions (VCs)}: a set of mathematical statements, usually of the form $A_i\sqsubseteq B_i$ (see Lemma \ref{lem-correctness}), is generated from the annotated program, where $A_i,B_i$ are matrices, and $\sqsubseteq$ stands for the L\"{o}wner order. It is required that if all the generated VCs are true, then $\vdash_\mathit{par} \{A\}P\{B\}$ (or $\vdash_\mathit{tot}\{A\}P\{B\}$). 
\item \textit{Proving VCs}: prove that matrix $B_i-A_i$ is positive semi-definite for each VC $A_i\sqsubseteq B_i$.  
\end{itemize}

Recently, a theorem prover was built by Liu, Li, Wang et al. in \cite{Liu} for the quantum Hoare logic based on the proof assistant Isabelle/HOL. It has been used to verify several quantum programs, including Grover search and phase estimation, which is a key in Shor's factoring algorithm. We can expect that more automatic tools for verification of quantum programs and quantum cryptographic protocols will be built after quantum computers be commercialised.

\section{Invariants of Quantum Programs}\label{invariant}

It has been well-understood since the very beginning that invariant generation is crucial for automatic verification of programs, and the problem of invariant generation has been intensively studied for classical programs. At this moment, invariants are provided by humans to the theorem prover for quantum programs developed in \cite{Liu} rather than generated by the system itself. But the issue of generating quantum invariants was recently considered in \cite{YYW17}. In this section, we briefly review the main results in \cite{YYW17}.
 
We first observed that the control flow of a quantum program can be represented by a super-operator-valued transition system (SVTS):
\begin{defn}[Super-Operator-Valued Transition Systems \cite{YYW17}] An SVTS is a $5$-tuple: $$\mathcal{S}=\langle\mathcal{H},L,l_0,\mathcal{T},\Theta\rangle,$$ where: \begin{enumerate}\item $\mathcal{H}$ is a Hilbert space; \item $L$ is a finite set of locations; \item $l_0\in L$ is the initial location; \item $\Theta$ is a quantum predicate in $\mathcal{H}$ denoting the initial condition; and \item $\mathcal{T}$ is a set of transitions. Each transition $\tau\in \mathcal{T}$ is written as $\tau=l\stackrel{\mathcal{E}}{\rightarrow}l^\prime$ with $l,l^\prime\in L$ and $\mathcal{E}$ being a super-operator in $\mathcal{H}$. For each $l\in L$, it is required that $$\mathcal{E}_l=\sum\{|\mathcal{E}: l\stackrel{\mathcal{E}}{\rightarrow} l^\prime\in\mathcal{T}|\}$$ 
is trace-preserving, i.e. $\mathit{tr}(\mathcal{E}_l(\rho))=\mathit{tr}(\rho)$ for all $\rho$.\end{enumerate}\end{defn}

Now we assume that the control flow SVTS of quantum program $P$ is $\mathcal{S}_P$ with state Hilbert space $\mathcal{H}_P$ (see \cite{YYW17} for detailed construction of $\mathcal{S}_P$). A set $\Pi$ of paths in $\mathcal{H}_P$ is said to be prime if for each $$\pi=l_1\stackrel{\mathcal{E}_1}{\rightarrow}...\stackrel{\mathcal{E}_{n-1}}{\rightarrow}l_n\in\Pi,$$ its proper initial segments $$l_1\stackrel{\mathcal{E}_1}{\rightarrow}...\stackrel{\mathcal{E}_{k-1}}{\rightarrow}l_k\notin\Pi$$ for all $k<n$. We write $\mathcal{E}_\pi$ for the composition of $\mathcal{E}_1,...,\mathcal{E}_{n-1}$ and $\mathcal{E}_\Pi=\sum\left\{|\mathcal{E}_\pi:\pi\in\Pi|\right\}$. Then we can define invariants for quantum programs:

\begin{defn}[Invariants of Quantum Programs \cite{YYW17}] Let $\mathcal{S}_P=\langle\mathcal{H}_P,L,l_0,\mathcal{T},\Theta\rangle$ be the control flow SVTS of quantum program $P$ and $l\in L$. An invariant at location $l\in L$ is a quantum predicate $O$ in $\mathcal{H}_P$ satisfying the condition: for any density operator $\rho$ and prime set $\Pi$ of paths from $l_0$ to $l$, we have: 
\begin{equation*}\tr(\Theta\rho)\leq 1- \tr\left(\mathcal{E}_\Pi(\rho)\right)+\tr\left(O\mathcal{E}_\Pi(\rho)\right).\end{equation*}\end{defn}

It was shown in \cite{YYW17} that invariants can be used to establish partial correctness of quantum programs. We can also introduce the notion of inductive assertion for quantum programs. It is not easy to identify the invariant in the proof rules (R.LP) and (R.LT) for quantum loops, at least not as explicit as in the Hoare rule for classical loops. Example 4.1 in \cite{YYW17} indicates that the quantum predicate $M_0^\dag AM_0+M_1^\dag BM_1$ in rules (R.LP) and (R.LT) can be viewed as an inductive assertion (invariant). Furthermore, by generalising the constraint-based technique of Col\'{o}n et al. \cite{Colon03}, \cite{SSM}, it was demonstrated in \cite{YYW17} that invariant generation for quantum programs can be reduced to an SDP (Semidefinite Programming) problem. In particular, the method proposed in \cite{YYW17} was actually applied to generate the invariants of the quantum walk in Example \ref{exam1} and quantum Metropolis sampling \cite{TOVPV11}. 

\section{Terminations Analysis of Quantum Programs}\label{termination}

As is well-known, termination analysis is a key step in proving total correctness of programs. Termination of quantum programs has been researched along the following two lines:\begin{itemize}\item Algorithmic analysis of termination for quantum programs was first considered in \cite{YF10} where the Jordan decomposition of complex matrices was employed as the main tool. In \cite{Yi13}, the author and his collaborators introduced quantum Markov chains as a semantic model of quantum programs. Then in a series of their papers \cite{YSG}, \cite{YuY}, \cite{Li-14}, termination of quantum programs was reduced to the reachability problem of quantum Markov chains, which is in turn tackled by developing a theory of quantum graphs (see \cite{Ying2016}, Section 5.2). Indeed, this line of research also paves a way to model-checking quantum systems.

\item The notion of ranking function was defined in \cite{Ying2011} in order to present the proof rule (R.LT) for total correctness of quantum loops. It was recalled as Definition \ref{bou-def} in this paper. In the last few years, (super)martingales have been employed as a powerful mathematical tools for termination analysis of probabilistic programs \cite{Chak13}, \cite{Fior15}, \cite{Chat16}. 
Recently, a notion of quantum (super)martingales was introduced in \cite{Ying2016b} as a generalisation of Definition \ref{bou-def}  and used to characterise the termination of quantum programs. The synthesis problem of (super)martingales for quantum programs was also investigated there. The basic idea is that the fundamental Gleason theorem \cite{Gleason} in quantum foundations can be used to determine the template of ranking functions, and then the synthesis problem of quantum (super)martingales with certain templates can also be reduced to an SDP problem. It seems that a further development of this line of research requires us to systematically establish a mathematical theory of quantum (super)martingales first.\end{itemize}

The termination of the quantum walk in Example \ref{exam1} was analysed by both of the above approaches \cite{Yi13}, \cite{Ying2016b}, and the second approach was also used in \cite{Ying2016b} to analyse the termination of quantum Bernoulli factory, a quantum algorithm for random number generation \cite{DJR15}.

\section{Conclusion}

We conclude this paper by pointing out several problems for future research:
\begin{itemize}
\item \textit{Combining Classical Computation and Quantum Computation}: We have been focusing on pure quantum programs without classical computation. This allows us to have a clean theory of quantum programming. But almost all existing quantum algorithms involves both quantum and classical computation, and the state-of-art quantum programming languages like Quipper \cite{Quipper}, LIQUi|> \cite{LIQUi}, Q\# \cite{Svor18}, Scaffold \cite{Scaffold} and QWire \cite{Pay17} include both classical and quantum variables. So, it is desirable to generalise the verification techniques discussed in this paper to quantum programs involving classical computation. Indeed, a kind of correctness formulas (Hoare triples) involving both classical and quantum variables were already introduced in \cite{YF11}, but a Hoare-like logic for programs with both classical and quantum variables is still to be developed. 

\item \textit{Parallel and Distributed Quantum Programs}: Distributed quantum computing has been studied for 20 years since \cite{Grover} and \cite{Cleve}, including finding quantum algorithms for solving paradigmatic problems from classical distributed computing \cite{Leader}, experiments toward physical implementation of distributed quantum computing \cite{SMB} and architecture of distributed quantum hardware systems \cite{Meter}. 
In particular, since practical quantum computers with large qubit capacity are still out of the current technology's reach, it is an attractive idea to use the physical resources of two or more small-capacity quantum computers to realise a large-capacity quantum computing system. 
So, another interesting problem is how the program logic and related techniques considered in this paper can be extended for reasoning about parallel and distributed quantum programs. The notion of proof outline for quantum programs introduced in Section \ref{outline} should be helpful in dealing with the (non)interference of the variables shared by component quantum programs.   

\item \textit{Verification of Quantum Cryptographic Protocols}: Over the last 10 years, quantum communication has blossomed into a viable commercial technology. But verifying correctness and security of quantum communication protocols is a notoriously difficult problem. Process algebra approach has been introduced for verification of quantum cryptographic protocols \cite{Gay}, \cite{FengY15}, \cite{Ku16}. On the other hand, a (probabilistic) relational Hoare logic (pRHL) and a machine-checked framework for reasoning about security and privacy in classical computing and communicating systems were established by Barthe, Fournet, K\"{o}pf et al. \cite{BFG}, \cite{BKO}. 
The success of this line of research suggests us to develop verification techniques for quantum cryptographic protocols by extending our quantum Hoare logic. Indeed, the first attempt to develop a quantum relational Hoare logic (qRHL) was already made by Unruh in \cite{Un18}, where a tool for the verification based on qRHL was implemented and successfully used to the security proof of several quantum cryptographic protocols.  
\end{itemize}

\textbf{Acknowledgment}: The author likes to thank Professors Martin Fr\"{a}nzle, Deepak Kapur and Naijun Zhan for inviting me to give a talk at SETTA'2016. 
This work was partly supported by the Australian Research Council (Grant No: DP160101652) and the Key Research Program of Frontier Sciences, Chinese Academy of Sciences.

\appendix
\section{Basic Properties of Operators in Hilbert Spaces}\label{app-a}

For convenience of the reader, we first review some basic properties of operators and super-operators needed in the proofs presented in Appendix \ref{app-b}. 
The following lemma gives a characterisation of the L\"{o}wner order in terms of trace. 

\begin{lem} Let $A, B$ be observables (i.e. Hermitian operators) in Hilbert space $\mathcal{H}$. Then $A\sqsubseteq B$ if and only if $\mathit{tr}(A\rho)\leq\mathit{tr}(B\rho)$ for all density operators in $\mathcal{H}$. 
\end{lem}

The notion of dual super-operator is introduced in the next definition. It will be extensively used in Appendix \ref{app-b}. 

\begin{defn}\label{def-dual} Let $\mathcal{E}$ be a quantum operation (i.e. super-operator) in Hilbert space $\mathcal{H}$ with the Kraus representation $\mathcal{E}=\sum_i E_i\circ E_i^\dag$. Then its (Schr\"{o}dinger-Heisenberg) dual is the super-operator $\mathcal{E}^\ast$ defined by $$\mathcal{E}^\ast(A)=\sum_i E_i^\dag AE_i$$ for any observable $A$ in $\mathcal{H}$.  
\end{defn}

The next lemma presents a characterisation of duality between super-operators in terms of trace. 

\begin{lem}\label{lem-dual}For any quantum operation $\mathcal{E}$, observable $A$ and density operator $\rho$ in $\mathcal{H}$, we have: $$\mathit{tr}(A\mathcal{E}(\rho))=\mathit{tr}(\mathcal{E}^\ast(A)\rho).$$ In particular, $\mathit{tr}(\mathcal{E}(\rho))=\mathit{tr}(\mathcal{E}^\ast(I)\rho)$, where $I$ is the identity operator in $\mathcal{H}$. 
\end{lem}

Positivity of operators and the L\"{o}wner order between operators in a tensor product of Hilbert spaces are considered in the following: 

\begin{lem}\label{lem-tensor} \begin{enumerate}\item If $A_1,A_2$ are positive operators in $\mathcal{H}_1$ and $\mathcal{H}_2$, respectively, then $A_1\otimes A_2$ is a positive operator in $\mathcal{H}_1\otimes\mathcal{H}_2$. 
\item For any operators $A_1,B_1$ in $\mathcal{H}_1$ and $A_2, B_2$ in $\mathcal{H}_2$, $A_1\sqsubseteq B_1$ and $A_2\sqsubseteq B_2$ implies $A_1\otimes A_2\sqsubseteq B_1\otimes B_2.$ 
\end{enumerate}
\end{lem}

The proofs of all of the above lemmas can be found in any standard textbook on the theory of operators in Hilbert spaces. 

\section{Compositions of Quantum Relations}\label{app-relation}

In Section \ref{sec-examples}, we already saw some simple quantum relations (i.e. quantum predicates in the tensor product of more than one state Hilbert spaces) in specifying correctness of quantum programs. 
The notion of composition of classical relations is well-defined and widely used to construct complicated relations from simple ones. However, it is highly nontrivial to find an appropriate definition of the notion of composition of quantum relations. 
Here, we give several tentative definitions: 

\begin{defn} Let $A$ be a quantum relation between $\mathcal{H}_1$ and $\mathcal{H}_2$ (i.e. an observable in $\mathcal{H}_1\otimes\mathcal{H}_2$ between the zero and identity operators) and $B$ a quantum relation between $\mathcal{H}_2$ and $\mathcal{H}_3$. We define three kinds of their compositions: 
\begin{enumerate}\item Circle composition: $$A\circ_\mathcal{B} B=\frac{1}{d_2}\sum_i\langle ii|A\otimes B|ii\rangle$$ where and in the sequel $d_2=\dim\mathcal{H}_2$ and $\mathcal{B}=\{|i\rangle\}$ is an orthonormal basis of $\mathcal{H}_2$.
\item Bullet composition: $$A\bullet_\mathcal{B} B=\langle\Psi|A\otimes B|\Psi\rangle$$ where $|\Psi\rangle_\mathcal{B}=\frac{1}{\sqrt{d_2}}\sum_i|ii\rangle$ is the (unnormalised) maximal entanglement defined by an orthonormal basis  $\mathcal{B}=\{|i\rangle\}$ of $\mathcal{H}_2$.
\item Diamond composition: $$A\diamond_v B=\mathit{tr}_{\mathcal{H}_2^{\otimes 2}}\left[S_v(A\otimes B)S_v\right]$$
where $v\in\{+,-\}$, $S_{\pm}$ are the symmetrisation and anti-symmetrisation operators (see Example \ref{ex-symmetry}), and $\mathit{tr}_{\mathcal{H}_2^{\otimes 2}}$ stands for tracing out the middle two $\mathcal{H}_2$'s in $\mathcal{H}_1\otimes\mathcal{H}_2\otimes\mathcal{H}_2\otimes\mathcal{H}_3$.\end{enumerate}
\end{defn}

The algebraic structure of quantum relations equipped with the above composition operations and corresponding transitive closures is a very interesting topic for future research.  

\section{Proofs of Theorems}\label{app-b}

In this section, we provide the proofs of theorems omitted in the main part of this paper.  

\subsection{Proof of Theorem \ref{thm.strong-sound}}\label{proof.strong-sound}

To prove the strong soundness, suppose that $$\langle P,\rho\rangle\rightarrow^n\{|\langle P_i,\rho_i\rangle|\}.$$ We proceed by induction on the length $n$ of computation. 

{\vskip 4pt}

\textbf{Induction Basis}: For $n=0$, $\{|\langle P_i,\rho_i\rangle|\}$ is a singleton $\{|\langle P_1,\rho_1\rangle|\}$ with $P_1\equiv P$ and $\rho_1\equiv \rho$. Then we can choose $T_1\equiv P$ and it holds that $P_1\equiv\mathit{at}(T_1,P)$. Note that in the proof outline $\{A\}P^\ast\{B\}$, we have $A\sqsubseteq\mathit{pre}(P)=B_1$. Thus, $$\mathit{tr}(A\rho)\leq\mathit{tr}(B_1\rho_1)=\sum_i\mathit{tr}(B_i\rho_i).$$ 

\textbf{Induction Step}: Now we assume that $$\langle P,\rho\rangle\rightarrow^{n-1}\{|\langle P_i,\rho_i\rangle|\}\rightarrow\{|\langle P_i,\rho_i\rangle|i\neq i_0|\}\cup\{|\langle Q_j,\sigma_j\rangle|\}$$ where $$\langle P_{i_0},\rho_{i_0}\rangle\rightarrow\{|\langle Q_j,\sigma_j\rangle|\}$$ is derived by one of the rules used in defining transition relation (\ref{op-sem-extend}). Then we need to consider the following cases: 

{\vskip 4pt}

$\blacktriangleright$ Case 1. The last step uses rule (IF$^\prime$). Then $P_{i_0}$ can be written in the following form: $$P_{i_0}\equiv\ \mathbf{if}\ \left(\square m\cdot M[\overline{q}]=m\rightarrow R_m\right)\ \mathbf{fi},$$ and for each $j$, $Q_j\equiv R_m\equiv \mathit{at}(R_m,P)$ and $\sigma_j=M_m\rho_{i_0}M_m^\dag$ for some $m$. On the other hand, a segment of the proof outline $\{A\}P^\ast\{B\}$ must be derived by the following inference:
$$\frac{\{A_m\}R_m^\ast\{C\}\ {\rm for\ every}\ m}{\left\{\sum_mM_m^\dag A_m M_m\right\}\ \mathbf{if}\ \left(\square m\cdot M[\overline{q}]=m\rightarrow\left\{A_m\right\} R_m^\ast\right)\ \mathbf{fi}\{C\}}$$ and $$B_{i_0}=\mathit{pre}\left(P_{i_0}\right)\sqsubseteq\sum_m M_m^\dag A_mM_m,\ \ \ \ A_m=\mathit{pre}\left(R_m\right).$$ Therefore, \begin{align*}\mathit{tr}\left(B_{i_0}\rho_{i_0}\right)&\leq\mathit{tr}\left(\sum_mM_m^\dag A_mM_m\rho_{i_0}\right)\\ 
&=\sum_m\mathit{tr}\left(M_m^\dag A_mM_m\rho_{i_0}\right)\\ &=\sum_m\mathit{tr}\left(A_mM_m\rho_{i_0}M_m^\dag\right)\\ &=\sum_j\mathit{tr}\left(\mathit{pre}\left(Q_j\right)\sigma_j\right).
\end{align*} By the induction hypothesis, we obtain: \begin{align*}\mathit{tr}(A\rho)&\leq\sum_{i\neq i_0}\mathit{tr}\left(B_i\rho_i\right)+\mathit{tr}\left(B_{i_0}\rho_{i_0}\right)\\ 
&\leq \sum_{i\neq i_0}\mathit{tr}\left(B_i\rho_i\right)+\sum_j\mathit{tr}\left(\mathit{pre}\left(Q_{j}\right)\sigma_{j}\right).
\end{align*} So, the conclusion is true in this case.

{\vskip 4pt}

$\blacktriangleright$ Case 2. The last step uses rule (L$^\prime$). Then $P_{i_0}$ must be in the following form: $$P_{i_0}\equiv\ \mathbf{while}\ M[\overline{q}]=1\ \mathbf{do}\ R\ \mathbf{od}$$ and $$\{|\langle Q_j,\sigma_j\rangle|\}=\{|\langle Q_0,\sigma_0\rangle,\langle Q_1,\sigma_1\rangle|\}$$ with $Q_0\equiv\mathbf{skip}, \sigma_0=M_0\rho_{i_0}M_0^\dag, Q_1\equiv\ R;P_{i_0}$ and $\sigma_1=M_1\rho_{i_0}M_1^\dag$. A segment of $\{A\}P^\ast\{B\}$ must be derived by the inference in Figure \ref{fig 4.0} \begin{figure*}[!t]\normalsize\begin{align*}\frac{\{D\}R^\ast\{M_0CM_0^\dag+M_1DM_1^\dag\}}{\begin{array}{ccc}\{M_0CM_0^\dag+M_1DM_1^\dag\}\ \mathbf{while}\ M[\overline{q}]=0\ \mathbf{do}\ \{C\}\ \mathbf{skip}\ \{C\}\\ 
\ \ \ \ \ \ \ \ \ \ \ \ \ \ \ \ \ \ \ \ \ \ \ \ \ \ \ \ \ \ \ \ \ \ \ \ \ \ \ \ \ \ \ \ \ \ \ \ \ \ \ \ \ \ \ \ \ \ \ \ \ \ \ \ \ \ \ \ \ \ \ \ \ \ \ \ \ \ \ \ \ \ \ \ \ \ \ \ \ \ \ \ \ \ =1\ \mathbf{do}\ R^\ast \{M_0CM_0^\dag+M_1DM_1^\dag\}\ \mathbf{od}\ \{C\}
\end{array}}\end{align*}\hrulefill \vspace*{4pt}
\caption{Proof of Theorem \ref{thm.strong-sound}}\label{fig 4.0}
\end{figure*} and $B_{i_0}\sqsubseteq M_0CM_0^\dag +M_1DM_1^\dag$. Then $Q_0\equiv\mathit{at}(\mathbf{skip},P), \mathit{pre}(Q_0)=C, Q_1\equiv\mathit{at}(R,P)$ and $\mathit{pre}(Q_1)=D$. It follows that \begin{align*}
\mathit{tr}\left(B_{i_0}\rho_{i_0}\right)&\leq\mathit{tr}\left[\left(M_0CM_0^\dag+M_1DM_1^\dag\right)\rho_{i_0}\right]\\ &=\mathit{tr}\left(M_0CM_0^\dag\rho_{i_0}\right) +\mathit{tr}\left(M_1DM_1^\dag\rho_{i_0}\right)\\
&=\mathit{tr}\left(CM_0^\dag\rho_{i_0}M_0\right)+\mathit{tr}\left(DM_1^\dag\rho_{i_0}M_1\right)\\ 
&= \mathit{tr}\left(\mathit{pre}(Q_0)\sigma_0\right)+\mathit{tr}\left(\mathit{pre}(Q_1)\sigma_1\right).
\end{align*} Furthermore, by the induction hypothesis, we have: \begin{align*}\mathit{tr}(A\rho)&\leq\sum_{i\neq i_0}\mathit{tr}\left(B_i\rho_i\right)+\mathit{tr}\left(B_{i_0}\rho_{i_0}\right)\\ 
&\leq \sum_{i\neq i_0}\mathit{tr}\left(B_i\rho_i\right)+\sum_j\mathit{tr}\left(\mathit{pre}\left(Q_{j}\right)\sigma_{j}\right).
\end{align*} Thus, the conclusion is true in this case.

{\vskip 4pt}

$\blacktriangleright$ Case 3. The last step uses rule (Sk), (In) or (UT). Similar but easier.   

\subsection{Proof of Theorem \ref{Aux-Sound}}\label{proof.1}

We only prove the theorem for partial correctness; the case of total correctness is simpler.

{\vskip 4pt}

1. We first prove that rule (Ax.Inv) is sound for partial correctness. Since $\mathit{var}(P)\cap V=\emptyset$, $\llbracket P\rrbracket$ can be seen as a super-operator $\mathcal{E}$ in $\mathcal{H}_{V^c}$. Then when considering $\llbracket P\rrbracket$ as a super-operator in $\mathcal{H}$, we have: \begin{align*}\llbracket P\rrbracket^\ast&(A)+\left[I-\llbracket P\rrbracket^\ast(I)\right]\\ &=B\otimes\mathcal{E}^\ast(I_{V^c})+\left[I-I_V\otimes\mathcal{E}^\ast\left(I_{V^c}\right)\right]\\ &=B\otimes\mathcal{E}^\ast\left(I_{V^c}\right)+I_V\otimes\left[I_{V^c}-\mathcal{E}^\ast\left(I_{V^c}\right)\right]\\ 
&\sqsupseteq B\otimes\mathcal{E}^\ast\left(I_{V^c}\right)+B\otimes\left[I_{V^c}-\mathcal{E}^\ast\left(I_{V^c}\right)\right]\\ 
&=B\otimes I_{V^c} = A
\end{align*} because $B\sqsubseteq I_{V^c}$ and $\mathcal{E}^\ast\left(I_{V^c}\right)\sqsubseteq I_{V^c}$, where $I$ stands for the identity operator in $\mathcal{H}$. 

{\vskip 4pt}

2. Now we prove the soundness of rule (R.TI) for partial correctness. Assume that $\models_\mathit{par}\{A\}P\{B\}$. Then it holds that $$A\sqsubseteq \llbracket P\rrbracket^\ast\left(B\otimes I_W\right)+\left(I-\llbracket P\rrbracket^\ast(I)\right)$$ where $I$ is the identity operator in $\mathcal{H}_{V\cup W}$.  
Note that $\mathit{var}(P)\subseteq V$. Then we have: $$\mathit{tr}_W\left(\llbracket P\rrbracket^\ast\left(B\otimes I_W\right)\right)=\llbracket P\rrbracket^\ast(B)$$ where the occurrences of $\llbracket P\rrbracket$ in the left-hand side and right-hand side are seen as a super-operator in $\mathcal{H}_{V\cup W}$ and one in $\mathcal{H}_V$, respectively. Similarly, we have  $\mathit{tr}_W\left(\llbracket P\rrbracket^\ast (I)\right)=\llbracket P\rrbracket^\ast(I_V).$
Therefore, it follows that 
 \begin{align*}
\mathit{tr}_WA &\sqsubseteq \mathit{tr}_W\llbracket P\rrbracket^\ast\left(B\otimes I_W\right)+\mathit{tr}_W\left(I-\llbracket P\rrbracket^\ast(I)\right)\\ 
&=\llbracket P\rrbracket^\ast(B) +\left(I_V-\llbracket P\rrbracket^\ast\left(I_V\right)\right)
\end{align*} and $\models_\mathit{par}\left\{\mathit{tr}_WA\right\}P\{B\}.$

{\vskip 4pt}
 
3. We prove the soundness of (R.CC) for partial correctness. If for every $i=1,...,m$, $\models_\mathit{par}\{A_i\}P\{B_i\}$, then by Lemma \ref{lem-correctness} we have: $$A_i\sqsubseteq \llbracket P\rrbracket^\ast(B_i)+\left[I-\llbracket P\rrbracket^\ast(I)\right].$$ Consequently, it holds that \begin{align*}\sum_{i=1}^mp_iA_i &\sqsubseteq\sum_{i=1}^mp_j\left(\llbracket P\rrbracket^\ast(A_j)+[I-\llbracket P\rrbracket^\ast(I)]\right)\\ &=\llbracket P\rrbracket^\ast\left(\sum_{i=1}^mp_i B_i\right)+[I-\llbracket P\rrbracket^\ast(I)]
\end{align*} because $\llbracket P\rrbracket^\ast(\cdot)$ is linear, and $I-\llbracket P\rrbracket^\ast(I)$ is positive. Thus, we have: $$\models_\mathit{par}\left\{\sum_{i=1}^mp_iA_i\right\}P\left\{\sum_{i=1}^mp_iB_i\right\}.$$

{\vskip 4pt}

4. We prove the soundness of (R.Inv) for partial correctness. From $\models_\mathit{par}\{A\}P\{B\}$, i.e. $A\sqsubseteq\llbracket P\rrbracket^\ast(B)+\left(I-\llbracket P\rrbracket^\ast(I)\right)$, we derive that \begin{align*}pA+qC&\sqsubseteq p\llbracket P\rrbracket^\ast(B)+p\left(I-\llbracket P\rrbracket^\ast(I)\right)\\
&=\left(p\llbracket P\rrbracket^\ast(B)+qC\right)+p\left(I-\llbracket P\rrbracket^\ast(I)\right)\\
&\sqsubseteq\left(p\llbracket P\rrbracket^\ast(B)+qC\right)+\left(I-\llbracket P\rrbracket^\ast(I)\right)\\
&=\llbracket P\rrbracket^\ast(pB+qC)+\left(I-\llbracket P\rrbracket^\ast(I)\right)
\end{align*} because $p\leq 1$ and $V\cap\mathit{var}(P)=\emptyset$ implies $\llbracket P\rrbracket^\ast(C)=C$. Therefore, we have $\models_\mathit{par}\{pA+qC\}P\{pB+qC\}$. 

{\vskip 4pt}

5. We prove that rule (R.SO) is sound for partial correctness. Suppose that $\models_\mathit{par}\{A\}P\{B\}$, i.e. $A\sqsubseteq\llbracket P\rrbracket^\ast(B)+(I-\llbracket P\rrbracket^\ast(I))$. Note that $\mathcal{E}$ is a super-operator in $\mathcal{H}_V$ and $V\cap\mathit{var}(P)=\emptyset$. Then $\mathcal{E}^\ast$ and $\llbracket P\rrbracket^\ast$ commutes, i.e. $\mathcal{E}^\ast\circ\llbracket P\rrbracket^\ast=\llbracket P\rrbracket^\ast\circ\mathcal{E}^\ast$. Moreover, it holds that $\mathcal{E}^\ast(I)\sqsubseteq I$ and thus $$\left(\mathcal{I}-\llbracket P\rrbracket^\ast\right)\left(\mathcal{E}^\ast(I)\right)\sqsubseteq \left(\mathcal{I}-\llbracket P\rrbracket^\ast\right)(I)$$ where $\mathcal{I}$ is the identity super-operator. Therefore, we obtain:\begin{align*}\mathcal{E}^\ast(A)&\sqsubseteq\mathcal{E}^\ast\left[\llbracket P\rrbracket^\ast(B)+(I-\llbracket P\rrbracket^\ast(I))\right]\\ 
&=\mathcal{E}^\ast\left(\llbracket P\rrbracket^\ast(B)\right)+\left[\mathcal{E}^\ast(I)-\mathcal{E}^\ast\left(\llbracket P\rrbracket^\ast(I)\right)\right]\\
&=\llbracket P\rrbracket^\ast\left(\mathcal{E}^\ast(B)\right)+\left(\mathcal{I}-\llbracket P\rrbracket^\ast\right)\left(\mathcal{E}^\ast(I)\right)\\
&\sqsubseteq\llbracket P\rrbracket^\ast\left(\mathcal{E}^\ast(B)\right)+\left(\mathcal{I}-\llbracket P\rrbracket^\ast\right)(I)\\
&=\llbracket P\rrbracket^\ast\left(\mathcal{E}^\ast(B)\right)+\left(I-\llbracket P\rrbracket^\ast(I)\right)
\end{align*} and $\models_\mathit{par}\{\mathcal{E}^\ast(A)\}P\{\mathcal{E}^\ast(B)\}$. 

{\vskip 4pt}

6. Finally, we prove the soundness of (R.Lim) for partial correctness. The existence of $\lim_{n\rightarrow\infty}A_n$ and $\lim_{n\rightarrow\infty}B_n$ is guaranteed by Lemma 4.1.3 in \cite{Ying2016}. Assume that $\models_\mathit{par}\{A_n\}P\{B_n\}$. Then $$A_n\sqsubseteq\llbracket P\rrbracket^\ast(B_n)+\left[I-\llbracket P\rrbracket^\ast(I)\right]$$ and the continuity of super-operator $\llbracket P\rrbracket^\ast$ yields: \begin{align*}
A =\lim_{n\rightarrow\infty}A_n&\sqsubseteq\lim_{n\rightarrow\infty}\left(\llbracket P\rrbracket^\ast(B_n)+\left[I-\llbracket P\rrbracket^\ast(I)\right]\right)\\
&=\llbracket P\rrbracket^\ast(\lim_{n\rightarrow\infty}B_n)+\left[I-\llbracket P\rrbracket^\ast(I)\right]\\
&=\llbracket P\rrbracket^\ast(B)+\left[I-\llbracket P\rrbracket^\ast(I)\right].
\end{align*}  Thus, $\models_\mathit{par}\{A\}P\{B\}$.
\end{document}